\documentclass[11pt]{article}
\usepackage{defs}

\title{The Greedy Superstring Conjecture Is False}
\author{
   Nikolai Chukhin$^{*}$
   \and
   Rishikesh Gajjala$^{\dagger}$
   \and
   Juha Harviainen$^{\ddagger}$
   \and
   Alexander S. Kulikov$^{*}$
   \and
   Ivan Mihajlin$^{*}$
   \and
   Hiroki Shibata$^{\mathsection}$
   \and
   Alexander V. Smal$^{*}$
}
\date{\small
   $^{*}$JetBrains Research\qquad
   $^{\dagger}$Pramaana Labs\qquad
   $^{\ddagger}$University of Helsinki\qquad
   $^{\mathsection}$Kyushu University
}

\begin{document}
\maketitle

\begin{abstract}
	In~the Shortest Common Superstring problem (SCS), one is~given a~finite set of~strings and is~asked to~find a~shortest string containing every input string as~a~substring.
	Since the problem is~NP-hard, various approximation algorithms have been proposed for~it.
	The currently best known approximation ratio is~$2.466$ due to~Englert, Matsakis, and Vesel{\'y};
	a~better approximation ratio~$2$ has been announced by~OpenAI just recently.
	
	In~1988, Tarhio and Ukkonen 
	introduced a~simple greedy algorithm that iteratively takes two strings with the longest overlap and merges them together.
	They showed that the greedy algorithm may produce a~superstring that is about two times longer than the optimal one
	and conjectured that this bound is~tight. This has become known as~the Greedy Superstring Conjecture and has remained open for almost four decades. The best known upper bound on~the approximation ratio of~the greedy algorithm 
	is~$3.396$ due to~Englert, Matsakis, and Vesel{\'y}.
	
	In this paper, we present a~series of~results tightening the bounds of~the greedy algorithm, with the main contributions being as follows:
	\begin{enumerate}
		\item We~refute the Greedy Conjecture by~constructing a~family of~instances for which a~lower bound on~the approximation ratio tends to~$2.4$.
		\item We prove that the greedy algorithm is~$3$-approximate.
		\item For input strings of~fixed length~$k$, we prove that the approximation guarantee is exactly $9/5$ for $k=3$ and 
		at~least~$2$ for all $k \ge 6$.
	\end{enumerate}
\end{abstract}

 \newpage

\section{Introduction}
\label{sec:introduction}

In~the Shortest Common Superstring problem (SCS), one is~given a~finite set of~strings and asked to~find a~shortest string containing every input string as~a~substring.
The problem is~a~classical model of~assembling highly overlapping data and appears in~data compression and genome assembly~\cite{GP14}, where one obtains short subsequences of the DNA and wants to reconstruct the whole sequence under the assumption that the shortest superstring is the most plausible true sequence.
SCS~is~\NP-hard already for strings of~length three~\cite{GMS80},
and, for this reason, polynomial-time approximation algorithms are actively studied.
The first constant-factor algorithm, with approximation ratio~$3$, was given by Blum et al.~\cite{BJLTY91}.
Subsequent work improved this ratio through increasingly delicate combinations of~cycle covers and maximum asymmetric traveling salesperson algorithms.
The best published approximation ratio is~$2.466$~\cite{EMV23}; recently, OpenAI announced a $2$-approximation~\cite{OpenAI26}.

There however exists also a simple, almost four-decade-old greedy algorithm that is conjectured \cite{TU88,Storer88,Turner89,BJLTY91}
to be equally good:
while more than one string remains, choose an~ordered pair with maximum suffix--prefix overlap, merge the pair, and iterate. 
The conjecture that the greedy algorithm is $2$-approximate has then become known as the Greedy Superstring Conjecture.
What makes this algorithm particularly appealing is that it admits a~linear-time implementation over a~constant-size alphabet~\cite{Ukkonen90} and is~known to~behave well in~practice~\cite{RBT04,Ma09,SVB23}.
Whereas it~is not difficult to~show that its approximation ratio is~at~least~2,
the best known upper bound is~$3.396$~\cite{EMV23}.
\Cref{figure:timeline} summarizes the advancements on bounding the approximation guarantee of the greedy algorithm and the best known algorithms for SCS.

For many other \NP-hard problems, the behavior of~a~natural greedy algorithm has been understood for decades.
Take, for instance, the classical Vertex Cover problem: the algorithm ``while there~is an~uncovered edge, add both of~its endpoints to~the cover'' has approximation ratio~$2$ (folklore, see~\cite[Algorithm~1.2]{Vazirani01}).
Similarly, for Set Cover, the algorithm ``while there are uncovered elements, add a~set covering the largest number of~them'' has approximation ratio $\ln n$~\cite{Johnson74}.
Finally, for List Scheduling, the algorithm ``while there are unscheduled jobs, assign the next job to~the least loaded machine'' has approximation ratio $2-1/\mu$ on~$\mu$ machines~\cite{Graham66}.

For this reason, it~is surprising that the approximation ratio of~such a~remarkably simple greedy algorithm for SCS is~still unknown.
The Greedy Conjecture has been open for almost four decades, and even in~the special case of~strings of~a~fixed length the behavior of~the greedy algorithm is~poorly understood.

\begin{figure}[t]
	\begin{center}
		\begin{tikzpicture}[x=.33cm, y=2.6cm]
			\tikzstyle{l} = [gray!20, thin]
			\tikzstyle{d} = [circle, inner sep=0mm, minimum size=.8mm, fill=black]
			\tikzstyle{h} = [circle, inner sep=0mm, minimum size=1.1mm, draw=black, fill=white, thin]
			\tikzstyle{v} = [rotate=90, anchor=east, inner sep=.4mm, font=\scriptsize]

			\foreach \year in {1988, ..., 2027}
			\draw[l] (\year-1988, 2) -- (\year-1988, 4);
			\foreach \year in {1988, 1995, 2000, ..., 2020, 2026}
			\node[above] at (\year-1988, 4) {\small \year};
			\foreach \ratio in {4, 3, 2} {
				\draw[l] (0, \ratio) -- (39, \ratio);
				\node[right, xshift=1mm] at (39, \ratio) {\small \ratio};
			}
			\node[left, xshift=-6mm, align=center, font=\small] at (0, 4) {greedy\\upper\\bound};
			\node[left, xshift=-6mm, align=center, font=\small] at (0, 3) {SCS\\upper\\bound};
			\node[left, xshift=-6mm, align=center, font=\small] at (0, 2) {greedy\\lower\\bound};

			\draw (1991-1988, 3) -| (1993-1988, 2.889) -| (1994-1988, 2.75) -| (1995-1988, 2.725) -| (1996-1988, 2.667) -| (1997-1988, 2.596) -| (1999-1988, 2.5) -- (2013-1988, 2.5) -| (2013-1988, 2.479) -| (2022-1988, 2.475) -| (2023-1988, 2.466);
			\draw (2023-1988, 2.466) -| (2025.7-1988, 7/3);
\foreach \year/\ratio/\label/\ref in {
				1991/3/3.00/BJLTY91,
				1993/2.889/2.89/TY93,
				1994/2.75/2.75/{CGPR94, KPS94, AS94},
				1995/2.725/2.73/AS95,
				1996/2.667/2.67/AS96,
				1997/2.596/2.60/BJJ97,
				1999/2.5/2.50/Sweedyk99,
				2003/2.5/2.50/KLSS03,
				2005/2.5/2.50/KS05,
				2012/2.5/2.50/PEZ12,
				2013/2.479/2.48/Mucha13,
				2022/2.475/2.48/EMV22,
				2023/2.466/2.47/EMV23} {
				\node[d] at (\year-1988, \ratio) {};
				\node[v] at (\year-1988, \ratio-.04) {\label};
				\node[below] at (\year-1988, 1.74) {\footnotesize \rotatebox{90}{\cite{\ref}}};
			}
			\node[d] at (1994-1988, 2.794) {};
			\node[d] at (1994-1988, 2.834) {};

			\draw (1991-1988, 4) -| (2005-1988, 3.5) -| (2022-1988, 3.425) -| (2023-1988, 3.396) -| (2026-1988, 3);
			\foreach \year/\ratio/\label in {
				1991/4/4.00,
				2005/3.5/3.50,
				2022/3.425/3.42,
				2023/3.396/3.40} {
				\node[d] at (\year-1988, \ratio) {};
				\node[v] at (\year-1988, \ratio-.04) {\label};
			}

			\draw (0, 2) -| (2026.3-1988, 2.4);
			\node[d] at (0, 2) {};
			\node[v] at (0, 1.96) {2.00};
			\node[below] at (0, 1.74) {\footnotesize \rotatebox{90}{\cite{TU88}}};

			\node[h] at (2026-1988, 3) {};
			\node[v, anchor=east] at (2026-1988, 2.97) {3.0};
			
			\node[h] at (2026.3-1988, 2.4) {};
			\node[v, xshift=5.5mm] at (2026.3-1988, 2.4) {2.4};
			
			\node[h] at (2025.7-1988, 7/3) {};
			\node[v, xshift=-.8mm] at (2025.7-1988, 7/3) {2.34};
\node[below] at (2026-1988, 1.74) {\footnotesize \rotatebox{90}{this paper}};
		\end{tikzpicture}
	\end{center}
	\caption{Timeline of~approximation ratios for SCS (middle curve), upper bounds on~the approximation ratio of~the greedy algorithm (top curve), and lower bounds on~it (bottom curve).
}
	\label{figure:timeline}
\end{figure}

\subsection*{Our Contributions}

Our first main contribution is to refute the Greedy Conjecture\footnote{Some preliminary versions of the results of our manuscript are contained in the preprints \cite{Chukhin26a}, \cite{Chukhin26b}, and \cite{Shibata26}.}.
We~construct a~family of~instances of~increasing size for which a~lower bound on~the approximation ratio of~the greedy algorithm tends arbitrarily close to~$2.4$.
For this construction depending on~the choices of~two parameters, the strings use a~binary alphabet and a~particular execution of~the greedy algorithm results in~a~superstring whose length is~a~function of~the two parameters.
Since the worst-case approximation ratio of~the greedy algorithm is~independent of~the tie-breaking rule~\cite{Nikolaev21}, this lower bound holds for every tie-breaking rule.
We then give an explicit construction for a string containing all the input strings as~a~substring.
Although this is not necessarily the shortest superstring, it upper bounds the optimal solution and thus yields a lower bound for the approximation ratio.
Increasing the values of~parameters then brings the lower bound towards~$2.4$.
The strings themselves contain only four or~five ones, which are separated from each other by~gaps of~varying numbers of~consecutive zeros.
These gaps are carefully designed to enforce which strings get merged together such that afterwards the greedy algorithm has only merges with overlap $0$ remaining, with the intermediate strings starting with zeros and ending with ones.

\begin{theorem}[Informal version of~\Cref{thm:greedy-lower-bound}]\label{thm:greedy-lower-bound-informal}
	For every $\varepsilon>0$, there~is a~set of~binary strings on~which the greedy algorithm may output a~common superstring of~length at~least $(12/5-\varepsilon)\operatorname{OPT}$.
\end{theorem}

To complement our lower bound, we also tighten the upper bound for the approximation guarantee of the greedy algorithm to $3$. For this purpose, we start from the framework of Blum et al.~\cite{BJLTY91} for proving upper bounds for SCS. The approach constructs the \emph{overlap graph} of the input strings, which is a complete directed graph (with self-loops) where the weight of an edge from a string to another tells the number of symbols by which they overlap. Solving SCS reduces to finding the maximum-weight Hamiltonian path in the overlap graph, since to minimize the length of the superstring we want to maximize the sum of overlaps in merging the strings. The greedy algorithm for SCS then constructs a Hamiltonian path by iteratively adding edges of maximum weight such that adding the edge does not create a cycle nor increase the in- or out-degrees of a vertex to more than one. The latter constraint corresponds to each string being merged at most once on its left and right side. Letting the greedy algorithm add also edges resulting in cycles yields an optimal cycle cover of the overlap graph. This cycle cover can then be turned into a Hamiltonian path by removing one edge from each cycle with minimum overlap and merging the paths together. Our key contribution for tightening the upper bound is tightening a certain inequality from the framework of Blum et al.~\cite{BJLTY91} to its limit by proving the exact value of a particular coefficient appearing in the proof. New approaches would thus be needed to push the upper bound further.

\begin{restatable}{corollary}{threeApprox}\label{cor:three-approx}
	The greedy algorithm for SCS is a $3$-approximation algorithm.
\end{restatable}

MGREEDY~\cite{BJLTY91} concatenates the strings obtained by~opening these cycles; TGREEDY instead applies the greedy algorithm to~those strings.
The original bound of~$4$ on~MGREEDY also yielded the first $4$-approximation guarantee for the greedy algorithm.
We~prove that MGREEDY has approximation ratio exactly~$3$ for every tie-breaking rule (\Cref{cor:mgreedy-tight}).
For TGREEDY, we~prove an~upper bound of~$5/2$ for every execution and, on~the same binary instances used for the greedy algorithm, a~lower bound of~$12/5$ for some choices of~ties in~both stages (\Cref{cor:tgreedy-upper,cor:tgreedy-lower}).

Combining \Cref{thm:periodic-budget} with the $2/3$-approximation algorithm for MAX-ATSP~\cite{KLSS03,PEZ12} gives the following corollary.
\begin{corollary}\label{cor:seven-thirds-approx}
	There~is a~polynomial-time $7/3$-approximation algorithm for SCS.
\end{corollary}

Finally, we consider the performance of the greedy algorithm on strings of fixed length, asking if the performance improves on shorter strings.
Let $\rho_k$ be~the approximation ratio for strings of~length~$k$, that~is, the supremum of~$|g|/\operatorname{OPT}$ over all sets of~strings of~length~$k$ and all outputs~$g$ of~the greedy algorithm on~them.
Cazaux and Rivals~\cite{CR18} proved that $2-1/k\le\rho_k\le(k+1)/2$, which gives $\rho_2=3/2$, and they conjectured that $\rho_k=2-1/k$ for every~$k$.
For $k=4$, Kulikov, Savinov, and Sluzhaev~\cite{KSS15} proved that the greedy algorithm is~a~$2$-approximation if~it breaks ties in~favor of~strings of~the form $aaaa$, whereas for arbitrary tie-breaking the bounds of~\cite{CR18} give $7/4\le\rho_4\le5/2$.
We~prove that $\rho_3=9/5$ and $\rho_k\ge2$ for $k\ge6$, which disproves the conjecture of~Cazaux and Rivals for $k=3$ and for every $k\ge6$.
Together with \Cref{cor:three-approx}, this gives $2\le\rho_k\le3$ for every $k\ge6$.

For $k=3$, we~determine $\rho_3$ exactly, which gives the first exact value of~$\rho_k$ for any $k\ge3$.
For the lower bound $\rho_3\ge9/5$, we~use $\ell$ blocks of~five strings; in~the de~Bruijn graph on~strings of~length two, the strings of~each block form two cycles with two and three edges that share a~vertex, and all blocks share a~common vertex.
Hence a~single path visits all strings, whereas the greedy algorithm may merge the strings of~each cycle separately, after which no~two remaining strings overlap.

The proof for the upper bound is more involved: first, the greedy algorithm merges strings with overlap $2$, after which the remaining overlaps are of $1$ symbol, and we analyze these two phases separately. For overlaps of length $2$, we again consider the de Bruijn graph to obtain its trail cover through the merges. We then consider the overlap-one merges by studying a multigraph constructed from the trail cover. Finally, we analyze the types of trails to bound the length of the greedy output.

\begin{theorem}[Informal version of~\Cref{thm:length-three} and~\Cref{prop:length-three-tight}]\label{thm:length-three-informal}
	The approximation ratio of~the greedy algorithm on~strings of~length three is~$\rho_3=9/5$.
\end{theorem}

Similarly, for every $k\ge6$, our instance consists of~$\ell$ blocks of~$2k-3$ strings of~length~$k$; in~the de~Bruijn graph on~strings of~length $k-1$, the strings of~each block form two cycles with $k$ and $k-3$ edges that share a~vertex, and all blocks share a~common vertex.
Hence a~single path visits all $n=(2k-3)\ell$ strings with overlaps $k-1$, which gives a~shortest common superstring of~length $n+k-1$.
The greedy algorithm, however, may merge the strings of~each cycle separately, after which the remaining strings overlap by~at~most one symbol, and it~outputs a~superstring of~length~$2n$.

\begin{theorem}[Informal version of~\Cref{thm:fixed-length-lower-bound}]\label{thm:fixed-length-lower-bound-informal}
	For every $k\ge6$, the approximation ratio of~the greedy algorithm on~strings of~length~$k$ is~at~least~$2$, that~is, $\rho_k\ge2$.
\end{theorem}

The initial proofs for the above results were obtained by using GPT-5.6 Sol Pro, GPT-6 Astra, and Claude Opus-5.5. We have rewritten them with the goal of simplifying them and providing more intuition. See \hyperref[sec:disclaimer]{AI methodology} for further details. \section{Known Technical Results}
\label{sec:preliminaries}
We~start by~listing important known results on~SCS; we introduce
related concepts along the way.
Gray boxes are used to~visualize these concepts and facts;
they are intended to~help the reader and may~be safely skipped.

In~the rest of~the paper, $\mathcal{S}=\{s_1, \dotsc, s_n\}$ is~an~input to~SCS.
We~assume that $\mathcal{S}$ is~\emph{substring-free}, that is, no~string in~$\mathcal{S}$ is~a~substring of~another string in~$\mathcal{S}$.
Strings contained in~other input strings can be~removed during preprocessing.
By~$\operatorname{OPT}$ we~denote an~optimal common
superstring of~$\mathcal{S}$ or~its length.

\subsection{Prefixes and Overlaps}
For nonempty strings~$s$ and~$t$,
their \emph{overlap} $\operatorname{ov}(s, t)$
is~the longest (possibly empty) string~$q$ such that $s=p \circ q$ and $t=q \circ r$,
for some \emph{nonempty} strings $p$~and~$r$ that
are in~turn denoted by~$\operatorname{pref}(s,t)$ and
$\operatorname{suff}(s, t)$. For example,
$\operatorname{ov}({\tt ababa}, {\tt babc})={\tt ba}$
and
$\operatorname{ov}({\tt ababa}, {\tt ababa})={\tt aba}$.
We~call $\operatorname{ov}(s,s)$ the \emph{self-overlap} of~$s$.
By~$\langle s, t \rangle$ we~denote the \emph{overlap merge} of~$s$ and~$t$:
\[
	\langle s, t \rangle=
	s \circ \operatorname{suff}(s, t)=
	\operatorname{pref}(s, t) \circ t =
	\operatorname{pref}(s, t) \circ \operatorname{ov}(s, t) \circ \operatorname{suff}(s, t).
\]
Although the~overlap and the~prefix are defined as strings, we use these two words to denote the~lengths of the~corresponding strings.
The notion of~the overlap merge extends to~more than two strings
in~a~natural way:
for $t_1, \dotsc, t_k$,
\[
	\langle t_1, \dotsc, t_k \rangle=
	\operatorname{pref}(t_1, t_2)
	\circ
	\operatorname{pref}(t_2, t_3)
	\circ
	\dotsb
	\circ
	\operatorname{pref}(t_{k-1}, t_k)
	\circ
	t_k.
\]
Clearly,
\[
	|\langle t_1, \dotsc, t_k \rangle| =
	\sum_{i \in [k-1]}|\!\operatorname{pref}(t_i, t_{i+1})|+|t_k|=
	\sum_{i \in [k]}|t_i| - \sum_{i \in [k-1]}|\!\operatorname{ov}(t_i, t_{i+1})|.
\]
Hence, the smaller the sum of~adjacent prefixes, the larger
the sum of~adjacent overlaps.

\subsection{Permutations and Related Graphs}
The goal in~the SCS problem is~to find a~permutation of~$\mathcal{S}$
whose overlap merge length is~minimum. This~is equivalent to~maximizing
the total overlap or~minimizing the total prefix (up~to
the length of~the last string in the permutation).
The greedy algorithm tries to~construct a~permutation with large total overlap by~taking, at~every step, the largest available overlap.
Formally, it~maintains a~set of~\emph{current strings}, which is~initially~$\mathcal{S}$.
While more than one current string remains, it~chooses an~ordered pair of~distinct current strings with the longest overlap, breaking ties arbitrarily, and replaces them by~their overlap merge.
The last remaining current string is~the output.

\begin{graybox}
	For a~set of~strings
	\(\{{\tt abcdab}, {\tt cdabc}, {\tt dabcd}, {\tt ababa}, {\tt aaa}, {\tt babab}\}\),
	a~naive superstring results from concatenating them and has length
	$6+5+5+5+3+5=29$:
	\({\tt abcdabcdabcdabcdababaaaababab}\).
	An~optimal superstring has length~$14$:
	\texttt{aaabababcdabcd}.
	It~corresponds to~a~permutation
	\((\texttt{aaa}
	,
	\texttt{ababa}
	,
	\texttt{babab}
	,
	\texttt{abcdab}
	,
	\texttt{cdabc}
	,
	\texttt{dabcd})\).
	Its total overlap is~$1+4+2+4+4=15$, exactly the number of~symbols saved relative to~the naive concatenation of~length~$29$.

	The largest overlap is~$4$; therefore, the greedy algorithm may start by~merging $\langle \texttt{dabcd}, \texttt{abcdab}, \texttt{cdabc} \rangle$, after which it~may merge $\langle \texttt{babab}, \texttt{ababa}\rangle$.
	Afterwards, there are no~overlaps larger than~$1$, and the algorithm may finish with the following superstring of~length~$16$:
	\[
		\langle \texttt{babab}, \texttt{ababa}, \texttt{aaa}, \texttt{dabcd}, \texttt{abcdab}, \texttt{cdabc}\rangle = \texttt{bababaaadabcdabc}.
	\]
\end{graybox}

All pairwise overlaps between input strings can~be conveniently represented
in~an~\emph{overlap graph}:
it~is a~complete directed
graph with vertices~$\mathcal{S}$ and edges $\mathcal{S} \times \mathcal{S}$ (thus, self-loops are included) where the length of~an~edge $(s, t)$ is $|\!\operatorname{ov}(s, t)|$. A~\emph{prefix graph}
is~defined similarly with the only difference that the lengths are $|\!\operatorname{pref}(s, t)|$.

Solving SCS is~the same as~solving MAX-ATSP in~the overlap graph, that~is,
finding a~Hamiltonian path of~maximum total length. At~the same time, approximation algorithms for MAX-ATSP do~not give approximation guarantees for SCS directly since the total overlap may be~much larger
than $\operatorname{OPT}$.

\begin{graybox}
	The overlap and prefix graphs for our running example are~shown below. In~the overlap graph, a~Hamiltonian path of~maximum
	total overlap~15 is~highlighted.
	\begin{center}
		\begin{tikzpicture}
			\begin{scope}
				\foreach \s [count=\n] in {abcdab, cdabc, dabcd, ababa, aaa, babab}
					\node[v] (\s) at (60 * \n:20mm) {\strut \tt \s};

				\path (abcdab) edge[le, loop right, gray] node[l] {\scriptsize 2} (abcdab);
				\path (cdabc) edge[le, loop left, gray] node[l] {\scriptsize 1} (cdabc);
				\path (dabcd) edge[le, loop left, gray] node[l] {\scriptsize 1} (dabcd);
				\path (ababa) edge[le, loop left, gray] node[l] {\scriptsize 3} (ababa);
				\path (aaa) edge[le, loop right, gray] node[l] {\scriptsize 2} (aaa);
				\path (babab) edge[le, loop right, gray] node[l] {\scriptsize 3} (babab);

				\foreach \f/\t/\p/\o in {abcdab/cdabc/2/4,
					abcdab/dabcd/3/3,
					abcdab/ababa/4/2,
					abcdab/aaa/6/0,
					abcdab/babab/5/1,
					cdabc/abcdab/2/3,
					cdabc/dabcd/1/4,
					cdabc/ababa/5/0,
					cdabc/aaa/5/0,
					cdabc/babab/5/0,
					dabcd/abcdab/1/4,
					dabcd/cdabc/3/2,
					dabcd/ababa/5/0,
					dabcd/aaa/5/0,
					dabcd/babab/5/0,
					ababa/abcdab/4/1,
					ababa/cdabc/5/0,
					ababa/dabcd/5/0,
					ababa/aaa/4/1,
					ababa/babab/1/4,
					aaa/abcdab/2/1,
					aaa/cdabc/3/0,
					aaa/dabcd/3/0,
					aaa/ababa/2/1,
					aaa/babab/3/0,
					babab/abcdab/3/2,
					babab/cdabc/5/0,
					babab/dabcd/5/0,
					babab/ababa/1/4,
					babab/aaa/5/0}
				\path (\f) edge[e, gray] node[l] {\scriptsize \o} (\t);

				\foreach \f/\t/\o in {aaa/ababa/1, ababa/babab/4, babab/abcdab/2, abcdab/cdabc/4, cdabc/dabcd/4}
					\path (\f) edge[e, color=black, thick] node[l, color=black, fill=gray!25] {\scriptsize \o} (\t);
			\end{scope}

			\begin{scope}[xshift=80mm]
				\foreach \s [count=\n] in {abcdab, cdabc, dabcd, ababa, aaa, babab}
					\node[v] (\s) at (60 * \n:20mm) {\strut \tt \s};

				\path (abcdab) edge[le, loop right, gray] node[l] {\scriptsize 4} (abcdab);
				\path (cdabc) edge[le, loop left, gray] node[l] {\scriptsize 4} (cdabc);
				\path (dabcd) edge[le, loop left, gray] node[l] {\scriptsize 4} (dabcd);
				\path (ababa) edge[le, loop left, gray] node[l] {\scriptsize 2} (ababa);
				\path (aaa) edge[le, loop right, gray] node[l] {\scriptsize 1} (aaa);
				\path (babab) edge[le, loop right, gray] node[l] {\scriptsize 2} (babab);

					\foreach \f/\t/\p/\o in {abcdab/cdabc/2/4,
						abcdab/dabcd/3/3,
						abcdab/ababa/4/2,
						abcdab/aaa/6/0,
						abcdab/babab/5/1,
						cdabc/abcdab/2/3,
						cdabc/dabcd/1/4,
						cdabc/ababa/5/0,
						cdabc/aaa/5/0,
						cdabc/babab/5/0,
						dabcd/abcdab/1/4,
						dabcd/cdabc/3/2,
						dabcd/ababa/5/0,
						dabcd/aaa/5/0,
						dabcd/babab/5/0,
						ababa/abcdab/4/1,
						ababa/cdabc/5/0,
						ababa/dabcd/5/0,
						ababa/aaa/4/1,
						ababa/babab/1/4,
						aaa/abcdab/2/1,
						aaa/cdabc/3/0,
						aaa/dabcd/3/0,
						aaa/ababa/2/1,
						aaa/babab/3/0,
						babab/abcdab/3/2,
						babab/cdabc/5/0,
						babab/dabcd/5/0,
						babab/ababa/1/4,
						babab/aaa/5/0}
					\path (\f) edge[e, gray] node[l] {\scriptsize \p} (\t);

			\end{scope}
		\end{tikzpicture}
	\end{center}
\end{graybox}

The greedy superstring algorithm can be~conveniently represented in~terms
of~the overlap graph: it~constructs a~Hamiltonian path by~scanning
the list~of~edges in~the order of~nonincreasing length and takes the current edge into
a~solution if~it does not create a~node of~in-degree or~out-degree more than one and
does not close a~cycle. This way, at~any point of~the execution of~the greedy algorithm,
its solution is a~collection of~disjoint paths.

\begin{graybox}
	The greedy algorithm starts by~considering edges of~length~$4$.
	It~may take edges ${\tt dabcd} \to {\tt abcdab}$ and ${\tt abcdab} \to {\tt cdabc}$.
	Then, it~skips the edge ${\tt cdabc} \to {\tt dabcd}$ as~it~would create a~cycle.
	The greedy algorithm~then takes an~edge ${\tt babab} \to {\tt ababa}$. Finally, by~skipping a~bunch of~edges,
	it~takes edges ${\tt ababa} \to {\tt aaa}$ and ${\tt aaa} \to {\tt dabcd}$.

	\begin{center}
		\begin{tikzpicture}
			\foreach \s [count=\n] in {abcdab, cdabc, dabcd, ababa, aaa, babab}
				\node[v] (\s) at (60 * \n:18mm) {\strut \tt \s};

			\foreach \f/\t/\o in {dabcd/abcdab/4, abcdab/cdabc/4, babab/ababa/4, ababa/aaa/1, aaa/dabcd/0}
				\path (\f) edge[e, color=black, thick] node[l, color=black, fill=gray!25] {\scriptsize \o} (\t);
		\end{tikzpicture}
	\end{center}
\end{graybox}

This representation relies on~the following fact, which we~will use repeatedly in~our analysis of~the greedy algorithm.
\begin{lemma}[{\cite{BJLTY94}}]\label{lem:greedy-ends}
	Throughout any execution of~the greedy algorithm, the current strings remain substring-free.
	Moreover, if~two distinct current strings are formed by~merging the input strings $t_1, \dotsc, t_k$ and $t'_1, \dotsc, t'_\ell$, respectively, in~these orders, then
	\[\operatorname{ov}(\langle t_1, \dotsc, t_k\rangle, \langle t'_1, \dotsc, t'_\ell\rangle) = \operatorname{ov}(t_k, t'_1).\]
	Consequently, the overlap lengths selected by~the greedy algorithm are nonincreasing.
\end{lemma}

The \emph{compression} of~a~common superstring of~$\mathcal{S}$ is~the total length of~the strings in~$\mathcal{S}$ minus the length of~the superstring.
We~use the following classical result.

\begin{theorem}[\cite{TU88,Turner89,KN24}]\label{thm:half-compression}
	The compression of~every output of~the greedy algorithm on~$\mathcal{S}$ is~at~least half of~the compression of~every common superstring of~$\mathcal{S}$.
\end{theorem}

\subsection{Cycle Covers}
If~one allows the greedy algorithm to~produce cycles
(without violating the in-degree and out-degree constraints), then
it~constructs an~optimal \emph{cycle cover}~$\mathcal{C}$, that~is, a~collection of~cycles
that cover every node of~the graph and has maximum total overlap.
This is~proved in~\cite[Theorem~10]{BJLTY94}.
The constructed cycle cover may contain self-loops.
It~is not difficult to~see that $\mathcal{C}$~is also a~cycle cover
of~minimum total length in~the prefix graph.
Indeed, consider a~cycle $C=(t_0, t_1, \dotsc, t_{r-1}, t_0)$ of~$\mathcal{C}$.
Its total overlap is~equal
to~$\sum_{i \in [r]} |\!\operatorname{ov}(t_{i-1}, t_{i})|$
and its total prefix
is~$\sum_{i \in [r]} |\!\operatorname{pref}(t_{i-1}, t_{i})|$
(here, indices are taken modulo~$r$).
Note, however, that
$|\!\operatorname{pref}(t_{i-1}, t_{i})|+|\!\operatorname{ov}(t_{i-1}, t_{i})|=|t_{i-1}|$.
Consequently, the total overlap of~$C$ and its total prefix sum to~the total length of~the strings on~$C$.
Thus, maximizing total overlap minimizes total prefix.

For $C \in \mathcal{C}$, let $w(C)$ be~the total prefix of~the cycle~$C$ and $W=\sum_{C \in \mathcal{C}}w(C)$~be the total prefix of~$\mathcal{C}$. Then,
\begin{equation}
	W \le \operatorname{OPT}.
	\label{eq:wopt}
\end{equation}
Indeed, let $(t_1, \dotsc, t_n)$ be~a~permutation of~the input strings whose overlap merge has length~$\operatorname{OPT}$. Then, the total prefix of~the cycle
$(t_1, \dotsc, t_n, t_1)$ is~at~most $\operatorname{OPT}$, but at~least~$W$
(since this single cycle is also a~cycle cover):
\[\operatorname{OPT}=\sum_{i \in [n-1]}|\!\operatorname{pref}(t_i, t_{i+1})|+|t_n| \ge \sum_{i \in [n-1]}|\!\operatorname{pref}(t_i, t_{i+1})|+|\!\operatorname{pref}(t_n, t_1)| \ge W.\]

\begin{graybox}
	The figure below shows the same optimal cycle cover
	in~the overlap and prefix graphs.
	Its total overlap and total prefix are~$22$ and~$W = 7$, respectively; their sum is~$29$, the total length of~the input strings.
	\begin{center}
		\begin{tikzpicture}
			\begin{scope}
				\foreach \s [count=\n] in {abcdab, cdabc, dabcd, ababa, aaa, babab}
				\node[v] (\s) at (60 * \n:18mm) {\strut \tt \s};

				\foreach \f/\t/\o in {abcdab/cdabc/4, cdabc/dabcd/4, dabcd/abcdab/4, babab/ababa/4, ababa/babab/4}
				\path (\f) edge[e, black] node[l, fill=fc] {\scriptsize \o} (\t);

				\path (aaa) edge[le, loop right] node[l] {\scriptsize 2} (aaa);
			\end{scope}

			\begin{scope}[xshift=80mm]
				\foreach \s [count=\n] in {abcdab, cdabc, dabcd, ababa, aaa, babab}
				\node[v] (\s) at (60 * \n:18mm) {\strut \tt \s};

				\foreach \f/\t/\o in {abcdab/cdabc/2, cdabc/dabcd/1, dabcd/abcdab/1, babab/ababa/1, ababa/babab/1}
				\path (\f) edge[e] node[l] {\scriptsize \o} (\t);

				\path (aaa) edge[le, loop right] node[l] {\scriptsize 1} (aaa);
			\end{scope}
		\end{tikzpicture}
	\end{center}
\end{graybox}

\subsection{Opening the Cycles}
To~turn a~cycle cover into a~Hamiltonian path, one may proceed as~follows: for each cycle, remove its edge of~minimum overlap; then, merge the resulting paths into a~Hamiltonian path.

Formally,
consider a~cycle $C=(t_1, \dotsc, t_{r}, t_1)$ of~$\mathcal{C}$
and assume that the edge $(t_{r}, t_1)$ was the last edge of~this cycle added by~the algorithm.
We~call~it a~\emph{cycle-closing edge}; it~has minimum overlap among the edges of~the cycle.
Denote its length by $o(C)$ and let $O=\sum_{C \in \mathcal{C}}o(C)$. Now, let
\(\sigma(C)=\langle t_1, \dotsc, t_{r}\rangle\) be~a~\emph{representative string} for~$C$.
Then,
\[
	\sigma(C)=
	\operatorname{pref}(t_1, t_2)
	\circ
	\dotsb
	\circ
	\operatorname{pref}(t_{r-1}, t_r)
	\circ
	t_r
	=
	\operatorname{pref}(t_1, t_2)
	\circ
	\dotsb
	\circ
	\operatorname{pref}(t_{r-1}, t_r)
	\circ
	\operatorname{pref}(t_{r}, t_1)
	\circ
	\operatorname{ov}(t_{r}, t_1).
\]
Hence,
\[|\sigma(C)|=w(C)+o(C).\]
The algorithm MGREEDY~\cite{BJLTY91} constructs the cycle cover greedily and concatenates its representative strings.
More precisely, it~maintains a~set of~current strings, initially~$\mathcal{S}$, and an~initially empty set of~representatives.
While the current set is~nonempty, it~selects an~ordered pair of~current strings with maximum overlap, allowing the two selected strings to~coincide and breaking ties arbitrarily.
If~the strings are distinct, it~replaces them by~their overlap merge; otherwise, it~removes the selected string from the current set and adds it~unchanged to~the representatives.
The extracted strings are precisely the strings $\sigma(C)$ obtained by~opening the cycles at~their cycle-closing edges.
MGREEDY concatenates the representatives.
TGREEDY uses this construction of~the representatives as~its first stage and, as~its second stage, applies the greedy algorithm to~them, again breaking ties arbitrarily.
The superstring returned by~MGREEDY has length
\begin{equation}
	\sum_{C \in \mathcal{C}}(w(C)+o(C))=W+O.
	\label{eq:wo}
\end{equation}

\begin{graybox}
	The cycle-closing edges of~the optimal cycle cover are shown as~dashed edges below.
	Their total overlap is~$O=4+4+2=10$.
	To~the right, we~show the corresponding representative strings.
	By~concatenating the representative strings, one gets a~superstring
	{\tt abcdabcdbababaaaa}
	of~length $W+O=7+10=17$.

	\begin{center}
		\begin{tikzpicture}
				\foreach \s [count=\n] in {abcdab, cdabc, dabcd, ababa, aaa, babab}
					\node[v] (\s) at (60 * \n:18mm) {\strut \tt \s};

				\foreach \f/\t/\o in {abcdab/cdabc/4, cdabc/dabcd/4, babab/ababa/4}
					\path (\f) edge[e, black] node[l, fill=fc] {\scriptsize \o} (\t);

				\foreach \f/\t/\o in {dabcd/abcdab/4, ababa/babab/4}
					\path (\f) edge[e, black, dashed] node[l, fill=fc] {\scriptsize \o} (\t);

				\path (aaa) edge[le, loop right, dashed] node[l] {\scriptsize 2} (aaa);

				\node[right, text width=60mm] at (3, 0) {
					\begin{align*}
						\sigma(C_1) &= {\tt abcdabcd}\\
						\sigma(C_2) &= {\tt bababa}\\
						\sigma(C_3) &= {\tt aaa}\\
					\end{align*}
				};
		\end{tikzpicture}
	\end{center}
\end{graybox}

\subsection{Upper Bounding the Overlap of~Cycle-Closing Edges}
Since we~already know that $W \le \operatorname{OPT}$ (see \eqref{eq:wopt}),
to~turn an~upper bound $W+O$ (see \eqref{eq:wo})
into an~approximation guarantee,
we~need to~upper bound $O$~in~terms of~$\operatorname{OPT}$.
This is~exactly what many previous papers~do.
Namely, they prove an~upper bound of~the form
\begin{equation}
	O \le \operatorname{OPT} + \gamma W,
	\label{eq:o}
\end{equation}
for a~constant~$\gamma$.
Combined with~\eqref{eq:wopt} and~\eqref{eq:wo}, this inequality
immediately implies that MGREEDY is~a~$(2+\gamma)$-approximation.
As~shown by~\cite[Section 5]{BJLTY94}, this implies that the approximation guarantee of~the greedy algorithm is~also at~most $(2+\gamma)$.

The MGREEDY algorithm simply concatenates the representative strings.
If,~instead, one applies an~$\alpha$-approximation algorithm of~MAX-ATSP in~the overlap graph
of~the representative strings and merges them in~the found order,
the resulting algorithm has approximation guarantee $(2+(1-\alpha)\gamma)$.
This was shown in~\cite{BJJ97,Mucha07,Mucha13}; a~proof can also be~found in~\cite[Theorem 3.1]{EMV22}.
The history of~improvements of~upper bounds on~$\gamma$ is~the following:
\cite{BJLTY94} proved the inequality with $\gamma=2$,
\cite{KS05} improved~it to~$3/2$, \cite{EMV22} pushed it~to~$1.425$,
and \cite{EMV23} improved it~further to~$1.396$.

We~prove the inequality~\eqref{eq:o} for $\gamma=1$ (see \Cref{thm:periodic-budget}).
This implies that the greedy algorithm is~$3$-approximate.
Plugging~in the best known approximation guarantee $\alpha=2/3$~\cite{KLSS03, PEZ12} for MAX-ATSP gives a~polynomial-time $7/3$-approximation for SCS.
Our proof of~a~stronger upper bound on~$\gamma$ does not come at~a~cost
of a~more involved analysis: in~fact, our proof is~much simpler.
We~achieve this by~considering global properties of~the overlap graph
instead of~considering each cycle of~the optimal cycle cover separately.
Then, in~Theorem~\ref{prop:black-box-tightness}, we~prove that $\gamma=1$ is~the optimal endpoint of~this approach: for every $\gamma<1$, there exists
an~input to~SCS where the inequality~\eqref{eq:o} fails.

\subsection{Properties of~Cycle Strings}
A~cycle with small total prefix may~be an~indicator of~the high repetitive structure of~the strings of~this cycle. For example, a~cycle through strings
\[{\tt abcdabcdabcd},\, {\tt cdabcdabcdab},\, {\tt dabcdabcdabc}\]
has total prefix~$2 + 1 + 1 = 4$, whereas its minimum edge overlap is~$10$.
In particular, $O$ may be much larger than $W$.

Two strings are \emph{equivalent} if~one is~a~cyclic rotation of~the other.
A~string~$z$ is~\emph{primitive} if~there is~no~string~$v$ and an~integer $k\ge2$ such that $z=v^k$.
For a~cycle $C=(t_1,\dotsc,t_r,t_1)$ of~$\mathcal{C}$, define its \emph{cycle string} by
\[
	\tau(C)=\operatorname{pref}(t_1,t_2)\circ\cdots\circ\operatorname{pref}(t_r,t_1),
\] hence $|\tau(C)| = w(C)$.

\begin{lemma}[{\cite[Lemma 2.5]{Mucha13}}]\label{lem:primitive-cycle-words}
	Every cycle string~$\tau(C)$ is~primitive, and the cycle strings of~distinct cycles of~$\mathcal C$ are pairwise non-equivalent.
\end{lemma}

\begin{graybox}
	In~the running example, the three cycles~$C_1,C_2,C_3$ have the following cycle strings:
	\[
		\tau(C_1) =\texttt{abcd}, \qquad \tau(C_2) =\texttt{ba}, \qquad \tau(C_3) =\texttt{a}.
	\]
	Then, $\sigma(C_1) = \tau(C_1)^2, \sigma(C_2) = \tau(C_2)^3$, and $\sigma(C_3) = \tau(C_3)^3$ (in general, $\sigma(C)$ is~not always a~power of~$\tau(C)$).
\end{graybox}

Fix an~arbitrary total order on~the alphabet.
A~\emph{Lyndon word}~\cite{Shirshov53,Lyndon54} is~a~string that is~strictly smaller than each of~its proper suffixes; equivalently, it~is~primitive and lexicographically smaller than all its nontrivial cyclic rotations.
Every equivalence class of~primitive strings has a~unique Lyndon word.
Moreover, a~Lyndon word~$z$ has no~period $p<|z|$: otherwise its suffix of~length $|z|-p$ would equal a~proper prefix of~$z$ and hence would be~smaller than~$z$, contrary to~the definition.

For a~string~$x$ and an~integer $m\ge0$, let $\operatorname{Sub}_m(x)$ be~the set of~length-$m$ substrings of~$x$.
For a~string~$z$ and an~integer $m\ge0$, let $\operatorname{CSub}_m(z)$ be~the set of~length-$m$ substrings of~$z^\infty$.
For every cycle~$C\in \mathcal{C}$, let~$\tau'(C)$ be~the unique Lyndon word equivalent to~$\tau(C)$.

\begin{lemma}[{\cite[Claim 2]{BJLTY94}}]\label{lem:cyclic-substrings}
	For every cycle~$C\in \mathcal{C}$,
	\[\!\operatorname{CSub}_{o(C)+1}(\tau'(C))=\operatorname{CSub}_{o(C)+1}(\tau(C))\subseteq\operatorname{Sub}_{o(C)+1}(\operatorname{OPT}).\]
\end{lemma}

\begin{graybox}
	Consider the optimal superstring
	\(\operatorname{OPT}=\texttt{aaabababcdabcd}.\)
	Using the alphabet order $\texttt{a}<\texttt{b}<\texttt{c}<\texttt{d}$, the Lyndon words corresponding to~the three cycle strings are
	\[\tau'(C_1)=\texttt{abcd},\qquad \tau'(C_2)=\texttt{ab},\qquad \tau'(C_3)=\texttt{a}.\]
	The minimum overlaps of~the three cycles are $o(C_1)=o(C_2)=4$ and~$o(C_3)=2$.
	Then,
	\begin{align*}
		\operatorname{CSub}_{5}(\tau'(C_1))&= \{\texttt{abcda},\texttt{bcdab},\texttt{cdabc},\texttt{dabcd}\}, \\
		\operatorname{CSub}_{5}(\tau'(C_2))&=\{\texttt{ababa},\texttt{babab}\}, \\
		\operatorname{CSub}_{3}(\tau'(C_3))&=\{\texttt{aaa}\}.
	\end{align*}
	Thus, every length-$(o(C)+1)$ substring of~$\tau(C)^\infty$ occurs as~a~substring of~$\operatorname{OPT}$.
\end{graybox}

 \section{Optimal Bound on~the Cycle-Cover Inequality}\label{sec:periodic}
In~this section, we~give a~proof of~our main technical contribution: the smallest value of~$\gamma$ where the inequality~\eqref{eq:o} holds is~equal to~1.

\subsection{Upper Bound}

\begin{theorem}
	\label{thm:periodic-budget}
	For every instance of SCS, it holds that $O \le \operatorname{OPT}+W$.
\end{theorem}

\threeApprox*

\Cref{lem:cyclic-substrings} shows that, for every cycle~$C$, all length-$(o(C)+1)$ substrings of~$\tau(C)^\infty$ occur in~$\operatorname{OPT}$.
The lemma below bounds the number of~cycles that can use the same superstring.

\begin{lemma}\label{lem:substring-budget}
	Let~$u$ be~a~string, $m\in\{1,\ldots,|u|\}$, and let $z_1,\ldots,z_a$ be~distinct Lyndon words for some $a\ge0$.
	Suppose that $|z_j|\le m$ and $\operatorname{CSub}_m(z_j)\subseteq\operatorname{Sub}_m(u)$ for every~$j$.
	Then
	\[a\le|\!\operatorname{Sub}_m(u)|-|\!\operatorname{Sub}_{m-1}(u)|+1.\]
\end{lemma}

\begin{graybox}
	Consider the difference between $\operatorname{Sub}_{m}$ and $\operatorname{Sub}_{m - 1}$.
	Every distinct length-$(m-1)$ substring of~$u$, except possibly the one occurring only at~the end of~$u$, can~be extended to~the right by~one symbol.
	Therefore
	\[\!|\operatorname{Sub}_m(u)|\ge|\operatorname{Sub}_{m-1}(u)|-1.\]
	\Cref{lem:substring-budget} shows that each Lyndon word whose complete set of~length-$m$ periodic substrings occurs in~$u$ forces one additional length-$m$ substring.
	To~give an~example, let $u=\texttt{aababbabaab}$ and~$m=4$.
	Then
	\begin{align*}
		\operatorname{Sub}_3(u)&=\{\texttt{aab},\texttt{aba},\texttt{abb},\texttt{baa},\texttt{bab},\texttt{bba}\},\\
		\operatorname{Sub}_4(u)&=\{\texttt{aaba},\texttt{abaa},\texttt{abab},\texttt{abba}, \texttt{baab},\texttt{baba},\texttt{babb},\texttt{bbab}\}.
	\end{align*}
	By~the extension argument above, $|\operatorname{Sub}_4(u)| \ge |\operatorname{Sub}_3(u)|-1=5$.

	However, consider the Lyndon words~$z_1=\texttt{aab}$, $z_2=\texttt{ab}$, and~$z_3=\texttt{abb}$.
	Their length-$4$ cyclic substrings are
	\begin{align*}
		\operatorname{CSub}_4(z_1)&=\{\texttt{aaba},\texttt{abaa},\texttt{baab}\},\\
		\operatorname{CSub}_4(z_2)&=\{\texttt{abab},\texttt{baba}\},\\
		\operatorname{CSub}_4(z_3)&=\{\texttt{abba},\texttt{babb},\texttt{bbab}\}.
	\end{align*}
	They are contained in~$\operatorname{Sub}_4(u)$, so~\Cref{lem:substring-budget} implies that these three Lyndon words account for three additional length-$4$ substrings, which shows that
	\[
	|\operatorname{Sub}_{4}(u)| \ge a + |\operatorname{Sub}_3(u)| - 1 = 8.
	\]
\end{graybox}

\begin{proof}
Let~$G_m$ be~the following undirected multigraph: its vertices are~the strings in~$\operatorname{Sub}_{m-1}(u)$, and each string in~$\operatorname{Sub}_m(u)$ is~an~edge joining its length-$(m-1)$ prefix and its length-$(m-1)$ suffix (it~is~an~undirected version of~the order-$(m-1)$ Rauzy graph~\cite{Rauzy82} of~$u$, which~is~in~turn a~subgraph of~the de~Bruijn graph~\cite{DeBrujin46}).
	Thus, $G_m$ has $|\!\operatorname{Sub}_{m-1}(u)|$ vertices and $|\!\operatorname{Sub}_m(u)|$ edges.
	It~is~connected because the consecutive length-$(m-1)$ substrings of~$u$ form a~walk that visits every vertex.

	For every~$j$, let~$b_j$ be~the length-$m$ prefix of~$z_j^\infty$.
	Since $m\ge|z_j|$ and~$z_j$ is~a~Lyndon word, $b_j$ is~the unique lexicographically smallest string in~$\operatorname{CSub}_m(z_j)$.
	The strings~$b_1,\ldots,b_a$ are pairwise distinct.
	Indeed, suppose $b_i=b_j$.
	If~$|z_i|=|z_j|$, since $m\ge |z_i|$, the length-$|z_i|$ prefixes of~$b_i$ and~$b_j$ are~$z_i$ and~$z_j$, respectively, so~$z_i=z_j$.
	If, say, $|z_i|<|z_j|$, then $z_j$ is~a~prefix of~$z_i^\infty$ and hence has period $|z_i|<|z_j|$, which is~impossible for a~Lyndon word.
	After relabeling, assume $b_1<b_2<\cdots<b_a$ and delete the edges~$b_1,b_2,\ldots,b_a$ from~$G_m$ in~that order.
	We~show that after each deletion the graph remains connected.

	Consider the deletion of~the edge~$b_j$, and, for $0\le i<|z_j|$, let~$c_i$ and~$u_i$ be~the length-$m$ and length-$(m-1)$ substrings of~$z_j^\infty$ beginning at~offset~$i$, respectively.
	Then $c_0=b_j$, and each~$c_i$ is~an~edge of~$G_m$ connecting $u_i$ and~$u_{i+1}$.
	If~$|z_j|>1$, deleting the edge~$c_0 = b_j$ keeps the graph connected, since the alternate walk $c_1,c_2,\ldots,c_{|z_j|-1}$ is~still present: each $c_i$ with $1\le i<|z_j|$ is~larger than~$b_j$, whereas every previously deleted edge is~smaller than~$b_j$.
	If~$|z_j|=1$, then $b_j$ is~a~loop, so~its deletion is~harmless.

	\begin{graybox}
		For the same string $u=\texttt{aababbabaab}$ and~$m=4$, the graph~$G_4$ is~shown below.

		\begin{center}
			\begin{tikzpicture}
				\foreach \s/\x/\y in {aab/-32/10, aba/-13/0, baa/-32/-10, bab/13/0, abb/32/10, bba/32/-10}
					\node[v] (\s) at (\x mm,\y mm) {\strut \tt \s};

				\draw[dashed] (aab) -- node[l,pos=.5,above] {\scriptsize $b_1=\mathtt{aaba}$} (aba);
				\draw (aba) -- node[l,pos=.5,below] {\scriptsize $\mathtt{abaa}$} (baa);
				\draw (baa) -- node[l,pos=.5,above] {\scriptsize $\mathtt{baab}$} (aab);
				\draw[dashed] (aba) to[bend left=28] node[l,pos=.5,above] {\scriptsize $b_2=\mathtt{abab}$} (bab);
				\draw (aba) to[bend right=28] node[l,pos=.5,below] {\scriptsize $\mathtt{baba}$} (bab);
				\draw (bab) -- node[l,pos=.5,above] {\scriptsize $\mathtt{babb}$} (abb);
				\draw[dashed] (abb) -- node[l,pos=.5,above] {\scriptsize $b_3=\mathtt{abba}$} (bba);
				\draw (bba) -- node[l,pos=.5,below] {\scriptsize $\mathtt{bbab}$} (bab);
			\end{tikzpicture}
		\end{center}
	\end{graybox}

	After all $a$ deletions, the remaining graph is~connected and has $|\!\operatorname{Sub}_m(u)|-a$ edges; consequently, $|\!\operatorname{Sub}_m(u)|-a\ge|\!\operatorname{Sub}_{m-1}(u)|-1$.
\end{proof}

\begin{proof}[Proof of~\Cref{thm:periodic-budget}]
	Enumerate the cycles~$C\in \mathcal{C}$ satisfying $o(C)+1\ge w(C)$ as~$C_1,\ldots,C_h$.
	For every~index $j$, put $z_j=\tau'(C_j)$ and $\ell_j=o(C_j)+1$.
	The Lyndon words~$z_1,\ldots,z_h$ are distinct by~\Cref{lem:primitive-cycle-words}, and $|z_j|=|\tau(C_j)|=w(C_j)\le\ell_j$ for every~$j$.
	Moreover, \Cref{lem:cyclic-substrings} gives $\operatorname{CSub}_{\ell_j}(z_j)\subseteq\operatorname{Sub}_{\ell_j}(\operatorname{OPT})$ for every~$j$.
	This inclusion implies $\ell_j\le\operatorname{OPT}$, since $\operatorname{CSub}_{\ell_j}(z_j)$ is~nonempty, and $\operatorname{CSub}_m(z_j)\subseteq\operatorname{Sub}_m(\operatorname{OPT})$ for every $m\le\ell_j$, since every length-$m$ substring of~$z_j^\infty$ is~a~prefix of~a~length-$\ell_j$ one.
	For $1\le m\le \operatorname{OPT}$, let~$a_m$ be~the number of~indices~$j$ satisfying $|z_j|\le m\le\ell_j$.
	\Cref{lem:substring-budget} gives
	\[a_m\le|\!\operatorname{Sub}_m(\operatorname{OPT})|-|\!\operatorname{Sub}_{m-1}(\operatorname{OPT})|+1.\]
	Each index~$j$ is~counted by~$a_m$ for the values $m=|z_j|,\ldots,\ell_j$, hence exactly $\ell_j-|z_j|+1=o(C_j)-w(C_j)+2$ times.
	Since $|\!\operatorname{Sub}_0(\operatorname{OPT})|=|\!\operatorname{Sub}_{\operatorname{OPT}}(\operatorname{OPT})|=1$, summing the inequalities over~$m$ gives
	\begin{align*}
		\sum_{\substack{C\in \mathcal{C}\\o(C)+1\ge w(C)}}\bigl(o(C)-w(C)+2\bigr)&=\sum_{j=1}^{h}\bigl(\ell_j-|z_j|+1\bigr)\\
		&=\sum_{m=1}^{\operatorname{OPT}}a_m\\
		&\le\sum_{m=1}^{\operatorname{OPT}}\bigl(|\!\operatorname{Sub}_m(\operatorname{OPT})|-|\!\operatorname{Sub}_{m-1}(\operatorname{OPT})|+1\bigr)\\
		&=\operatorname{OPT}.
	\end{align*}
	For each cycle in~the first sum, $o(C)-w(C)\le o(C)-w(C)+2$, whereas every omitted cycle satisfies $o(C)-w(C)<0$.
	Therefore,
	\[
	O - W \le \operatorname{OPT}. \qedhere
	\]
\end{proof}

\subsection{Lower Bound}

In~this section, we~prove that the coefficient of~$W$ in~\Cref{thm:periodic-budget} cannot be~decreased further.

\begin{theorem}\label{prop:black-box-tightness}
	For every $\gamma<1$,
	there is~a~substring-free set~$\mathcal{S}$ of~binary strings
	such that
	\[O>\operatorname{OPT}+\gamma W.\]
\end{theorem}
\begin{proof}
	For $p\ge2$, put $s_p={\tt 0}^{p-1}{\tt 1}{\tt 0}^{p-1}{\tt 1}{\tt 0}^{p-2}$.
	\begin{graybox}
		\centering
		\(\begin{aligned}
			s_2 &= \texttt{0101}\\
			s_3 &= \texttt{0010010}\\
			s_4 &= \texttt{0001000100}\\
			s_5 &= \texttt{0000100001000}
		\end{aligned}\)
	\end{graybox}
	For $p\ne q$, the string~$s_p$ is~not a~substring of~$s_q$, since each string contains exactly two $\texttt{1}$'s, whose positions differ by~$p$ and~$q$, respectively.
	Furthermore,
	\[|\!\operatorname{ov}(s_p,s_p)|=2p-2\qquad\text{and}\qquad|\!\operatorname{ov}(s_p,s_q)|\le2p-3\quad\text{for }p\ne q.\]
	Indeed, $s_p$ starts and ends with $\texttt{0}^{p-1}\texttt{1}\texttt{0}^{p-2}$, giving a~self-overlap of~length at~least $2p-2$.
	An~overlap of~length $2p-2$ from $s_p$ to~$s_q$ forces $q=p$, because the first $\texttt{1}$ in~this suffix of~$s_p$ is~at~position~$p$, whereas the first $\texttt{1}$ in~$s_q$ is~at~position~$q$.
	A~longer proper overlap contains both $\texttt{1}$'s of~$s_p$, so~matching their distance forces $q=p$, and matching their positions then forces a~zero shift, which is~impossible for a~proper overlap.

	For $m\ge2$, let $\mathcal{S}_m=\{s_2,s_3,\dotsc,s_m\}$.
	By~the discussion above, $\mathcal{S}_m$ is~substring-free
	and its cycle cover consisting of~all self-loops has~minimum total prefix.
	Then,
	\begin{align*}
		W&=\sum_{p=2}^{m} (3 p - 2 - (2p - 2)) =\sum_{p = 2}^{m} p = m (m + 1) / 2 - 1, \\
		O&=\sum_{p=2}^{m}(2p-2)=2W-2(m-1).
	\end{align*}
	Since $|\!\operatorname{ov}(s_p,s_{p-1})|=2p-3$,
	\[
	\operatorname{OPT} \le
	|\langle s_m,s_{m-1},\dotsc,s_2\rangle| = |s_2|+\sum_{p=3}^{m}\bigl(|s_p|-(2p-3)\bigr)=4+\sum_{p=3}^{m}(p+1)=W+m.\]
	Consequently,
	\[O-\operatorname{OPT}\ge2W-2(m-1)-(W+m)=W-3m+2.\]
	Since $W=\Theta(m^2)$, for every fixed $\gamma<1$ and all sufficiently large~$m$, $O-\operatorname{OPT}>\gamma W$.
\end{proof}

\subsection{Consequences for MGREEDY and TGREEDY}

The cycle-cover inequality also gives the exact approximation ratio of~MGREEDY and an~improved upper bound for TGREEDY.

\begin{corollary}\label{cor:mgreedy-tight}
	For every tie-breaking rule, the approximation ratio of~MGREEDY is~exactly~$3$.
\end{corollary}
\begin{proof}
	Every output of~MGREEDY has length $W+O$, so~\Cref{thm:periodic-budget} and~\eqref{eq:wopt} give
	\[W+O\le\operatorname{OPT}+2W\le3\operatorname{OPT}.\]

	For the matching lower bound, consider the sets $\mathcal{S}_m=\{s_2,\dotsc,s_m\}$ from the proof of~\Cref{prop:black-box-tightness}.
	While $s_p$ is~the longest current string, $(s_p,s_p)$ is~the unique pair of~maximum overlap, with overlap $2p-2$, since every other pair has overlap at~most $2p-3$.
	Hence, for every tie-breaking rule, MGREEDY extracts $s_m,\dotsc,s_2$ unchanged, in~this order.
	Its final concatenation consequently has length
	\[|g_m|=\sum_{p=2}^{m}(3p-2)=3W-2(m-1),\]
	where $W=m(m+1)/2-1$.
	By~the proof of~\Cref{prop:black-box-tightness}, $\operatorname{OPT}\le W+m$, so
	\[\frac{|g_m|}{\operatorname{OPT}}\ge\frac{3W-2(m-1)}{W+m}=3-\frac{5m-2}{W+m}.\]
	The right-hand side tends to~$3$ as~$m\to\infty$, proving the corollary.
\end{proof}

\begin{corollary}\label{cor:tgreedy-upper}
	TGREEDY is~a~$5/2$-approximation algorithm.
\end{corollary}
\begin{proof}
	For a~set of~strings~$\mathcal{A}$, let $\operatorname{OPT}(\mathcal{A})$ denote the length of~a~shortest common superstring.
	Let $g$ be~any output of~TGREEDY on~$\mathcal{S}$.
	Let $\mathcal{R}=\{\sigma(C):C\in\mathcal{C}\}$ be~the set of~representative strings obtained in~the first stage, and let $L=\sum_{C\in\mathcal{C}}|\sigma(C)|=W+O$.
	The set $\mathcal{R}$ is~substring-free.
	Indeed, if~$\sigma(D)$ were a~substring of~$\sigma(C)$ for distinct cycles $C$ and~$D$, then every input string on~either cycle would occur in~$\tau(C)^\infty$ by~\cite[Claims~2 and~4]{BJLTY94}.
	By~\cite[Claim~3]{BJLTY94}, the vertices of~both cycles could then be~covered by~one cycle of~total prefix at~most $w(C)$.
	Replacing $C$ and~$D$ by~this cycle would decrease the total prefix of~$\mathcal{C}$ by~at~least $w(D)>0$, contradicting its optimality.

	We~next recall the bound
	\[\operatorname{OPT}(\mathcal{R})\le\operatorname{OPT}(\mathcal{S})+W,\]
	used in~the proof of~\cite[Theorem~11]{BJLTY94}.
	For every cycle $C=(t_1,\dotsc,t_r,t_1)$, let $s_C=t_1$ and $u_C=\langle t_1,\dotsc,t_r,t_1\rangle$.
	The string $u_C$ contains $\sigma(C)$, starts and ends with $s_C$, and has length $|s_C|+w(C)$.
	Choose an~ordering of~the strings $s_C$ whose overlap merge has minimum length.
	Merging the strings $u_C$ in~the same order preserves at~least the overlap of~each adjacent pair of~strings $s_C$, so~it~gives a~common superstring of~$\mathcal{R}$ of~length at~most $\operatorname{OPT}(\{s_C:C\in\mathcal{C}\})+W$.
	Since $\{s_C:C\in\mathcal{C}\}\subseteq\mathcal{S}$, this proves the bound.

	Since $\mathcal{R}$ is~substring-free, \Cref{thm:half-compression} applies to~the second stage of~TGREEDY and gives
	\[L-|g|\ge\frac12\bigl(L-\operatorname{OPT}(\mathcal{R})\bigr).\]
	Combining this with the bound on~$\operatorname{OPT}(\mathcal{R})$, \Cref{thm:periodic-budget}, and~\eqref{eq:wopt}, we~obtain
	\begin{align*}
		|g|&\le\frac12\bigl(L+\operatorname{OPT}(\mathcal{R})\bigr)\\
		&\le\operatorname{OPT}(\mathcal{S})+\frac32W\\
		&\le\frac52\operatorname{OPT}(\mathcal{S}).\qedhere
	\end{align*}
\end{proof}

 \section{Lower Bound for the Greedy Algorithm}
\label{sec:greedy-lower-bound}

In~this section, we~show that the approximation ratio of~the greedy algorithm is~at~least $12/5 = 2.4$.
Specifically, we~construct a~family of~sets of~binary strings for which the ratio of~the length of~the output to~that of~a~common superstring tends to~$12/5$ under an~execution that selects a~maximum overlap at~every step.
Since the worst-case approximation ratio of~the greedy algorithm is~independent of~the tie-breaking rule~\cite{Nikolaev21}, this lower bound holds for every tie-breaking rule.

Our main result is~as~follows:
\begin{theorem}[Formal version of~\Cref{thm:greedy-lower-bound-informal}]\label{thm:greedy-lower-bound}
	For all integers $r \ge 1$ and $P \ge r-1$, there is~a~set of~$5r$ binary strings that has a~common superstring of~length $5Pr+27r^2+4P+19r+1$, whereas the greedy algorithm may output a~common superstring of~length $12Pr+(115r^2+7r)/2$.
\end{theorem}

Therefore, the ratio of~the length of~the output to~$\operatorname{OPT}$ is~at~least
\[\frac{12Pr+(115r^2+7r)/2}{5Pr+27r^2+4P+19r+1}.\]
For $P=r^2$, this bound tends to~$12/5$ as~$r \to \infty$.

\subsection{Gap Words}

We~represent binary strings containing ones as~\emph{gap words}.
For nonnegative integers $k$, $L$, and $R$ and positive integers $g_1, \dotsc, g_k$, we~denote
\[[L \mid g_1\ g_2\ \cdots\ g_k \mid R] = {\tt 0}^{L}{\tt 1}\,{\tt 0}^{g_1-1}{\tt 1}\,{\tt 0}^{g_2-1}{\tt 1}\,\cdots\,{\tt 0}^{g_k-1}{\tt 1}\,{\tt 0}^{R}.\]
The sequence $(g_1, \dots, g_k)$ is~called a~\emph{gap sequence}.
For $k = 0$, the gap sequence is~empty and the gap word is~${\tt 0}^{L}{\tt 1}{\tt 0}^{R}$.
The gap word contains exactly $k+1$ ones.
For every $i \in \{1, \dotsc, k\}$, $g_i$ is~the distance between the $i$-th one and the $(i+1)$-st one.
We~call $g_1, \dotsc, g_k$ the \emph{gaps} and $L$ and $R$ the left and right \emph{pads}, respectively.

The overlap of~two gap words is~determined as~follows.
\begin{lemma}\label{lem:gap-word-overlap}
	Let $S = [L \mid u_1\ \cdots\ u_k \mid R]$ and $T = [L' \mid v_1\ \cdots\ v_\ell \mid R']$ be~gap words such that neither is~a~substring of~the other.
	Let $J$ be~the set of~integers $j$ with $0 \le j \le \min\{k, \ell\}$ that satisfy the following conditions:
	\begin{enumerate}
		\item $u_{k-j+i} = v_i$ for every $i \in \{1, \dotsc, j\}$.
		\item $L' < u_{k-j}$ (or~$L' \le L$ if~$j = k$).
		\item $R < v_{j+1}$ (or~$R \le R'$ if~$j = \ell$).
	\end{enumerate}
	Then
	\[
	\operatorname{ov}(S, T) = \begin{cases}
		[L' \mid v_1\ \cdots\ v_{j_{\max}} \mid R] & \text{if } J \ne \emptyset,\\
		{\tt 0}^{\min\{R, L'\}} & \text{otherwise},
	\end{cases}
	\]
	where $j_{\max} = \max J$ in~the first case.
\end{lemma}
\begin{proof}
	A~gap word that is~a~suffix of~$S$ and a~prefix of~$T$ must have the last $j$ gaps of~$S$ and the first $j$ gaps of~$T$ for some $0 \le j \le \min\{k, \ell\}$.
	Its left pad must be~$L'$ and its right pad must be~$R$, so~its gap word representation is~$[L' \mid v_1\ \cdots\ v_j \mid R]$.
	This gap word is~a~suffix of~$S$ and a~prefix of~$T$ precisely when the shared gaps agree and both pads fit into the neighboring gaps or~pads, as~specified by~the three conditions.
	Its length is~$L' + 1 + \sum_{i=1}^{j} v_i + R$, which increases strictly with $j$ because each $v_i$ is~positive.
	A~suffix of~$S$ that is~a~prefix of~$T$ and has no~gap word representation consists entirely of~zeros, and the longest such string is~${\tt 0}^{\min\{R, L'\}}$.
	Every gap word satisfying the three conditions is~longer than this string, so~$j_{\max}$ gives the longest common suffix and prefix when $J$ is~nonempty.
	Since neither $S$ nor $T$ is~a~substring of~the other, every common suffix and prefix is~proper in~both strings, so~the longest of~these strings is~$\operatorname{ov}(S, T)$.
\end{proof}

For the gap words $S$ and $T$ in~\Cref{lem:gap-word-overlap}, let $Q(S, T)$ be~the largest integer $q$ with $0 \le q \le \min\{k, \ell\}$ such that $u_{k-q+i}=v_i$ for every $i \in \{1, \dotsc, q\}$.
Since every $j \in J$ satisfies $j \le Q(S, T)$ and all gap values are positive, \Cref{lem:gap-word-overlap} gives
\begin{equation}
	|\operatorname{ov}(S, T)| \le L' + R + 1 + \sum_{i=1}^{Q(S, T)}v_i.
	\label{eq:gap-overlap-bound}
\end{equation}

\Cref{fig:gap-word-overlap} shows an~example of~gap words and their overlaps.

\begin{figure}[htbp]
	\centering
	\begin{tikzpicture}[x=10pt, y=1pt]
		\foreach \x in {24,27,31,34,40}
		\draw[gray!45, line width=.3pt] (\x,0) -- (\x,-119);
		\foreach \shift/\y/\kind in {21/0/R, 0/-34/L, 7/-68/L, 14/-102/L}
		{
		\begin{scope}[xshift=\shift*10pt, yshift=\y pt]
			\ifnum\y<0
			\pgfmathsetmacro{\matchstart}{20.5-\shift}
			\fill[gray!30] (\matchstart,4) rectangle (23.5,-19);
			\fi
			\node[left, inner sep=5pt, font=\small] at (-.5,0) {$\kind$};
			\ifnum\y<0
			\draw[gray, line width=.8pt] (-.5,0) -- (3,0) node[midway, above, inner sep=2pt, text=black, font=\scriptsize] {$3$};
			\draw[line width=1.3pt] (3,0) -- (20,0);
			\draw[gray, line width=.8pt] (20,0) -- (23.5,0) node[midway, above, inner sep=2pt, text=black, font=\scriptsize] {$3$};
			\foreach \x in {3,6,10,13,17,20}
			\fill (\x,0) circle (1.5pt);
			\foreach \a/\b/\gap in {3/6/3, 6/10/4, 10/13/3, 13/17/4, 17/20/3}
			\node[above, inner sep=2pt, font=\scriptsize] at ({(\a+\b)/2},0) {$\gap$};
			\foreach \x/\bit in {0/0,1/0,2/0,3/1,4/0,5/0,6/1,7/0,8/0,9/0,10/1,11/0,12/0,13/1,14/0,15/0,16/0,17/1,18/0,19/0,20/1,21/0,22/0,23/0}
			\node[font=\ttfamily\small, inner sep=0pt] at (\x,-13) {\bit};
			\else
			\draw[gray, line width=.8pt] (-.5,0) -- (3,0) node[midway, above, inner sep=2pt, text=black, font=\scriptsize] {$3$};
			\draw[line width=1.3pt] (3,0) -- (19,0);
			\draw[gray, line width=.8pt] (19,0) -- (22.5,0) node[midway, above, inner sep=2pt, text=black, font=\scriptsize] {$3$};
			\foreach \x in {3,6,10,13,19}
			\fill (\x,0) circle (1.5pt);
			\foreach \a/\b/\gap in {3/6/3, 6/10/4, 10/13/3, 13/19/6}
			\node[above, inner sep=2pt, font=\scriptsize] at ({(\a+\b)/2},0) {$\gap$};
			\foreach \x/\bit in {0/0,1/0,2/0,3/1,4/0,5/0,6/1,7/0,8/0,9/0,10/1,11/0,12/0,13/1,14/0,15/0,16/0,17/0,18/0,19/1,20/0,21/0,22/0}
			\node[font=\ttfamily\small, inner sep=0pt] at (\x,-13) {\bit};
			\fi
		\end{scope}
		}
	\end{tikzpicture}
	\caption{An~illustration of~the gap words $L = [3 \mid 3\ 4\ 3\ 4\ 3 \mid 3]$ and $R = [3 \mid 3\ 4\ 3\ 6 \mid 3]$ and three possible overlaps of~$L$ with $R$.
	Each dot on~a~horizontal line represents an~occurrence of~one in~the corresponding string.
	The numbers above the black and gray segments indicate the gap values and pad lengths, respectively.
	The corresponding binary string is~shown directly below each horizontal line.
	The overlapping portions are shaded.
	In~the three overlaps shown, the strings share only zeros, one gap, and three gaps, respectively.
	The longest overlap is~$\operatorname{ov}(L, R) = [3 \mid 3\ 4\ 3 \mid 3]$.}
	\label{fig:gap-word-overlap}
\end{figure}

\subsection{Construction and Overview}

We~now introduce the construction of~our lower-bound instances.
Fix integers $r \geq 1$ and $P \geq r-1$.
Let $M = P + 8r$ and $p_i = P + r + i$ for $1 \leq i \leq r$.
We~call $M$ the \emph{common gap} and $p_i$ the $i$-th \emph{private gap} for $1 \leq i \leq r$.
Let $\mathcal{S}$ consist of~the following five binary strings for each $i \in \{1, \dotsc, r\}$:

\begin{align*}
	X_i &= [P + 2r \mid p_i\ M\ p_i\ M \mid P], & Y_i &= [P \mid M\ p_i\ M\ p_i \mid P + 4r],\\
	U_i &= [P \mid M\ M\ p_i\ M \mid P], & V_i &= [P \mid M\ p_i\ M\ M \mid P + r],\\
	Z_i &= [P + r \mid M\ M\ p_i \mid 0].
\end{align*}
By~construction, the set $\mathcal{S}$ is substring-free.

\paragraph*{Overview}
We~now explain how the construction yields a~short common superstring and a~much longer string produced by~the greedy algorithm.
In~both constructions, we~first merge the five input strings for each $i$ into a~single string.
For the short common superstring and the long common superstring produced by~the greedy algorithm, we~write the strings obtained by~merging five strings as
\begin{align*}
	S_i &= \langle Z_i, U_i, Y_i, X_i, V_i \rangle,\\
	T_i &= \langle X_i, Y_i, U_i, V_i, Z_i \rangle,
\end{align*}
respectively.
We~choose the lengths of~the left and right pads of~the input strings so~that the greedy algorithm has an~execution that produces all the strings $T_i$.
As~shown in~\Cref{fig:greedy-prefix-comparison}, the sum of~the lengths of~the prefixes of~the four adjacent pairs is~$3P + O(r)$ for $S_i$ and $8P + O(r)$ for $T_i$.
Since the lengths of~the last strings are $|V_i| = 6P + O(r)$ and $|Z_i| = 4P+O(r)$, respectively, we~obtain $|S_i| = 9P + O(r)$ and $|T_i| = 12P + O(r)$.

We~next merge the strings obtained for different indices $i$.
For $1 \le i < r$, the strings $S_i$ and $S_{i+1}$ share a~long overlap of~length $4P + O(r)$, so~the length of~the corresponding prefix is~$5P + O(r)$.
In~contrast, $T_i$ and $T_j$ have empty overlap whenever $i \ne j$ because their private gaps differ and we~set the lengths of~their left and right pads to~prevent overlaps.
Thus, merging $S_1, \dotsc, S_r$ in~this order gives a~short common superstring, whereas merging $T_1, \dotsc, T_r$ gives their concatenation, with lengths
\begin{align*}
	|\langle S_1, \dotsc, S_r \rangle| &= \sum_{i=1}^{r-1}|\operatorname{pref}(S_i, S_{i+1})| + |S_r| = 5Pr + O(r^2 + P),\\
	|\langle T_1, \dotsc, T_r \rangle| &= \sum_{i=1}^{r}|T_i| = 12Pr + O(r^2).
\end{align*}
Consequently, the ratio of~the length of~the string produced by~the greedy algorithm to~the length of~the short common superstring approaches $12/5$ when we~choose $P=r^2$ and take the limit as~$r \to \infty$.

\begin{figure}[t]
	\centering
	\begin{tikzpicture}[x=1pt, y=1pt, line width=.65pt, >={Latex[length=6.5pt, width=6pt]}]
		\foreach \name/\angle in {Z/90, U/18, Y/-54, X/-126, V/-198}
		\node[v, minimum width=22pt, minimum height=16pt, font=\small] (short-\name) at (\angle:46pt) {$\name_i$};
		\foreach \from/\to/\value/\side in {Z/U/0/above, U/Y/1/above, Y/X/1/below, X/V/1/below}
		\draw[->] (short-\from) -- node[l, pos=.5, \side, fill=white, font=\scriptsize, inner sep=1.5pt] {$\value$} (short-\to);
		\draw[->, dash pattern=on 1.8pt off 1.6pt] (short-V) -- node[l, pos=.5, above, fill=white, font=\scriptsize, inner sep=1.5pt] {$2$} (short-Z);
		\begin{scope}[xshift=166pt]
			\foreach \name/\angle in {Z/90, U/18, Y/-54, X/-126, V/-198}
			\node[v, minimum width=22pt, minimum height=16pt, font=\small] (long-\name) at (\angle:46pt) {$\name_i$};
			\foreach \from/\to/\value/\side in {X/Y/1/above, Y/U/4/above, U/V/1/below, V/Z/2/below}
			\draw[->] (long-\from) -- node[l, pos=.5, \side, fill=white, font=\scriptsize, inner sep=1.5pt] {$\value$} (long-\to);
			\draw[->, dash pattern=on 1.8pt off 1.6pt] (long-Z) -- node[l, pos=.5, above, fill=white, font=\scriptsize, inner sep=1.5pt] {$4$} (long-X);
		\end{scope}
	\end{tikzpicture}
	\caption{Illustrations of~prefix graphs for the short common superstring (left) and the long string produced by~the greedy algorithm (right).
	An~integer label $x$ on~an~edge indicates that the length of~the corresponding prefix is~$Px+O(r)$.
	For $i < r$, the dashed arrows represent $V_i \to Z_{i+1}$ on~the left and $Z_i \to X_{i+1}$ on~the right, and their labels indicate the lengths of~the corresponding prefixes.}
	\label{fig:greedy-prefix-comparison}
\end{figure}

\subsection{Constructing a~Short Common Superstring}

We~first give an~explicit form of~the short common superstring described in~the overview and verify that it~contains every input string.
Let
\[H = [P + r \mid M\ M\ p_1\ M\ p_1\ \cdots\ M\ M\ p_r\ M\ p_r\ M\ M \mid P + r].\]
The gap sequence of~$H$ contains $3r+2$ occurrences of~$M$ and two occurrences of~each $p_i$.
The sum of~these gap values is~the distance between the first and last ones of~$H$.
Adding one for the first one and the lengths of~the two pads gives
\begin{align*}
	|H| &= (3r+2)M + 1 + 2(P+r) + 2\sum_{i=1}^{r}p_i\\
	&= (3r+2)(P+8r) + 1 + 2(P+r) + 2\sum_{i=1}^{r}(P+r+i)\\
	&= 5Pr + 27r^2 + 4P + 19r + 1.
\end{align*}

The gap sequence of~each string in~$\mathcal{S}$ appears as~a~consecutive segment of~that of~$H$.
For each such occurrence in~$H$, the runs of~zeros immediately before the first one and after the last one are at~least as~long as~the corresponding left and right pads, respectively~(see~\Cref{fig:host-occurrences}).
Thus, $H$ is~a~common superstring of~$\mathcal{S}$, which gives
\[\operatorname{OPT} \le |H| = 5Pr + 27r^2 + 4P + 19r + 1.\]

\begin{figure}[ht]
	\centering
	\begin{tikzpicture}[x=1pt, y=1pt]
		\foreach \j/\x in {0/0, 1/38, 2/86, 3/134, 4/172, 5/220, 6/258, 7/306, 8/354}
		\coordinate (h\j) at (\x, 0);
		\foreach \j in {1,...,8}
		\draw[gray!45, line width=.3pt] (h\j) -- ++(0,-153);
		\draw[densely dotted] (-10,0) -- (h0);
		\draw[thin] (h0) -- (384,0);
		\fill (384,0) circle (1.4pt);
		\draw[densely dotted] (384,0) -- (398,0);
		\node[left, font=\small] at (-10,0) {$H$};
		\foreach \j in {0,...,8}
		\fill (h\j) circle (1.4pt);
		\foreach \a/\b/\gap in {0/1/p_{i-1}, 1/2/M, 2/3/M, 3/4/p_i, 4/5/M, 5/6/p_i, 6/7/M, 7/8/M}
		\node[above, inner sep=2pt, font=\scriptsize] at ($(h\a)!.5!(h\b)$) {$\gap$};
		\foreach \name/\first/\last/\y/\left/\right/\L/\R in {
		S_i/1/8/-22/20/20/{P+r}/{P+r},
		Z_i/1/4/-44/20/0/{P+r}/0,
		U_i/1/5/-66/18/18/P/P,
		Y_i/2/6/-88/18/26/P/{P+4r},
		X_i/3/7/-110/22/18/{P+2r}/P,
		V_i/4/8/-132/18/20/P/{P+r}}
		{
		\draw[gray, line width=.8pt] ($(h\first)+(0,\y)$) -- ++(-\left,0) node[midway, above, inner sep=2pt, text=black, font=\scriptsize] {$\L$};
		\ifnum\right>0
		\draw[gray, line width=.8pt] ($(h\last)+(0,\y)$) -- ++(\right,0) node[midway, above, inner sep=2pt, text=black, font=\scriptsize] {$\R$};
		\fi
		\draw[line width=1.3pt] ($(h\first)+(0,\y)$) -- ($(h\last)+(0,\y)$);
		\foreach \j in {\first,...,\last}
		\fill ($(h\j)+(0,\y)$) circle (1.5pt);
		\node[right, inner sep=4pt, font=\small] at ($(h\last)+(\right,\y)$) {$\name$};
		}
		\node[above, inner sep=2pt, font=\scriptsize] at (371,0) {$p_{i+1}$};
		\draw[gray, line width=.8pt] (238,-154) -- (258,-154) node[midway, above, inner sep=2pt, text=black, font=\scriptsize] {$P+r$};
		\draw[line width=1.3pt] (258,-154) -- (384,-154);
		\draw[densely dotted, line width=1.3pt] (384,-154) -- (398,-154);
		\foreach \j in {6,...,8}
		\fill ($(h\j)+(0,-154)$) circle (1.5pt);
		\node[right, inner sep=4pt, font=\small] at (398,-154) {$S_{i+1}$};
		\begin{scope}[yshift=-198pt]
			\foreach \j/\x in {0/12, 1/44, 2/81.6667, 3/113.6667, 4/151.3333, 5/183.3333, 6/221, 7/258.6667, 8/290.6667, 9/328.3333, 10/366, 11/398}
			\coordinate (t\j) at (\x, 0);
			\foreach \j in {0,...,11}
			\draw[gray!45, line width=.3pt] (t\j) -- ++(0,-111);
			\draw[gray, line width=.8pt] (-10,0) -- (t0) node[midway, above, inner sep=2pt, text=black, font=\scriptsize] {$P+2r$};
			\draw[thin] (t0) -- (t11);
			\node[left, font=\small] at (-10,0) {$T_i$};
			\foreach \j in {0,...,11}
			\fill (t\j) circle (1.4pt);
			\foreach \a/\b/\gap in {0/1/p_i, 1/2/M, 2/3/p_i, 3/4/M, 4/5/p_i, 5/6/M, 6/7/M, 7/8/p_i, 8/9/M, 9/10/M, 10/11/p_i}
			\node[above, inner sep=2pt, font=\scriptsize] at ($(t\a)!.5!(t\b)$) {$\gap$};
			\foreach \name/\first/\last/\y/\left/\right/\L/\R in {
			X_i/0/4/-22/22/18/{P+2r}/P,
			Y_i/1/5/-44/18/26/P/{P+4r},
			U_i/5/9/-66/18/18/P/P,
			V_i/6/10/-88/18/20/P/{P+r},
			Z_i/8/11/-110/20/0/{P+r}/0}
			{
			\draw[gray, line width=.8pt] ($(t\first)+(0,\y)$) -- ++(-\left,0) node[midway, above, inner sep=2pt, text=black, font=\scriptsize] {$\L$};
			\ifnum\right>0
			\draw[gray, line width=.8pt] ($(t\last)+(0,\y)$) -- ++(\right,0) node[midway, above, inner sep=2pt, text=black, font=\scriptsize] {$\R$};
			\fi
			\draw[line width=1.3pt] ($(t\first)+(0,\y)$) -- ($(t\last)+(0,\y)$);
			\foreach \j in {\first,...,\last}
			\fill ($(t\j)+(0,\y)$) circle (1.5pt);
			\node[right, inner sep=4pt, font=\small] at ($(t\last)+(\right,\y)$) {$\name$};
			}
		\end{scope}
	\end{tikzpicture}
	\caption{An~illustration of~the occurrences of~strings in~$H$ and $T_i$.
	The top diagram shows $S_i$, $S_{i+1}$, and the five input strings with index $i$ in~$H$, and the bottom diagram shows the five input strings in~$T_i$.
	Each dot on~a~horizontal line represents an~occurrence of~one.
	The numbers above the black and gray segments indicate the gap values and pad lengths, respectively.}
	\label{fig:host-occurrences}
\end{figure}

\subsection{Constructing a~Long Common Superstring}

We~now introduce how the greedy algorithm can produce a~long common superstring of~$\mathcal{S}$.

For each $i \in \{1, \dotsc, r\}$, we~write the intermediate strings obtained as~we~merge the input strings as~follows.
\begin{align*}
	A_i &= \langle X_i, Y_i \rangle = [P+2r \mid p_i\ M\ p_i\ M\ p_i \mid P+4r],\\
	B_i &= \langle U_i, V_i \rangle = [P \mid M\ M\ p_i\ M\ M \mid P+r],\\
	C_i &= \langle B_i, Z_i \rangle = \langle U_i, V_i, Z_i \rangle = [P \mid M\ M\ p_i\ M\ M\ p_i \mid 0],\\
	T_i &= \langle A_i, C_i \rangle = \langle X_i, Y_i, U_i, V_i, Z_i \rangle\\
	&= [P+2r \mid p_i\ M\ p_i\ M\ p_i\ M\ M\ p_i\ M\ M\ p_i \mid 0].
\end{align*}
We~will show that the following procedure is~a~valid execution of~the greedy algorithm.
\begin{description}
	\item[Step~(i)] For $i = r, r-1, \dotsc, 1$, we~first merge $X_i$ with $Y_i$ to~obtain $A_i$ and then merge $U_i$ with $V_i$ to~obtain $B_i$.
	\item[Step~(ii)] For each $i \in \{1, \dotsc, r\}$, we~merge $B_i$ with $Z_i$ to~obtain $C_i$.
	\item[Step~(iii)] For each $i \in \{1, \dotsc, r\}$, we~merge $A_i$ with $C_i$ to~obtain $T_i$.
	\item[Step~(iv)] We~concatenate $T_1, \dotsc, T_r$ without overlap.
\end{description}

The overlaps used in~these merge procedures are
\begin{align*}
	\operatorname{ov}(X_i, Y_i) &= \operatorname{ov}(U_i, V_i) = [P \mid M\ p_i\ M \mid P],\\
	\operatorname{ov}(B_i, Z_i) &= [P+r \mid M\ M \mid P+r],\\
	\operatorname{ov}(A_i, C_i) &= [P \mid {} \mid P+4r],
\end{align*}
and their lengths are
\begin{align*}
	\alpha_i &= |\operatorname{ov}(X_i, Y_i)| = |\operatorname{ov}(U_i, V_i)| = 2P+2M+p_i+1 \\
	&= 5P+17r+i+1,\\
	\beta &= |\operatorname{ov}(B_i, Z_i)| = 2P+2r+2M+1 
	\\&= 4P+18r+1,\\
	\gamma &= |\operatorname{ov}(A_i, C_i)| \\
	&= 2P+4r+1.
\end{align*}
Since $P \geq r - 1$, we~have $\alpha_{r} \geq \cdots \geq \alpha_1 \geq \beta > \gamma$.
The overlap between $T_i$ and $T_j$ is~empty for $i \ne j$.
The correctness of~these overlaps follows from \Cref{lem:gap-word-overlap} and the definitions of~the strings.
The bottom diagram of~\Cref{fig:host-occurrences} illustrates the occurrences of~the five input strings in~$T_i$.

Since the gap sequence of~$T_i$ contains six occurrences of~$M$ and five occurrences of~$p_i$, we~have
\[|T_i| = 1+P+2r+6M+5p_i = 12P+55r+5i+1.\]
The strings $T_1, \dotsc, T_r$ are concatenated without overlap, so~the length of~the resulting string is
\[|T_1 \cdots T_r| = \sum_{i=1}^{r}(12P+55r+5i+1) = 12Pr+\frac{115r^2+7r}{2}.\]
Thus, it~remains to~show that this procedure is~a~valid execution of~the greedy algorithm.
Together with the bound $\operatorname{OPT} \le |H|$, this will establish the bound in~\Cref{thm:greedy-lower-bound}:
\[\frac{|T_1 \cdots T_r|}{\operatorname{OPT}} \ge \frac{|T_1 \cdots T_r|}{|H|} = \frac{12Pr+(115r^2+7r)/2}{5Pr+27r^2+4P+19r+1}.\]

\subsection{Validity of~the Greedy Procedure}

We~show that the procedure in~Steps~(i)--(iv) is~a~valid execution of~the greedy algorithm.
Throughout this subsection, we~use without further mention the overlap formula in~\Cref{lem:greedy-ends} and the fact that the maximum overlap length never increases.

We~first state a~lemma listing the possible pairs of~input strings with an~overlap longer than $\beta$.

\begin{lemma}\label{lem:overlaps-above-beta}
	If~a~pair of~distinct strings $S, T \in \mathcal{S}$ satisfies $|\operatorname{ov}(S, T)| > \beta$, then
	\[(S, T) \in \{(X_i, Y_i), (X_i, V_i), (U_i, Y_i), (U_i, V_i), (Y_i, X_i)\}\]
	holds for some $i \in \{1, \dotsc, r\}$.
\end{lemma}
\begin{proof}
	For every string in~$\mathcal{S}$, the left and right pads have lengths at~most $P+2r$ and $P+4r$, respectively, and every gap has value at~most $M$.
	If~$Q(S, T) \le 1$, then the inequality in~\eqref{eq:gap-overlap-bound} gives
	\[|\operatorname{ov}(S, T)| \le (P+2r) + (P+4r) + M + 1 = 2P+6r+M+1 = 3P+14r+1 < \beta.\]
	The gap sequences cannot share four or~more entries, so~$Q(S, T) \leq 3$.
	\Cref{tab:input-gap-sequences} lists the first and last two or~three entries of~the gap sequences of~the input strings.

	Suppose first that $Q(S, T)=2$.
	By~\Cref{tab:input-gap-sequences}, the two shared gaps are $(p_i, M)$, $(M, p_i)$, or~$(M, M)$.
	In~the first two cases, $p_i \le P+2r$, so~the inequality in~\eqref{eq:gap-overlap-bound} gives
	\[|\operatorname{ov}(S, T)| \le (P+2r)+(P+4r)+M+p_i+1 \le 4P+16r+1 < \beta.\]
	For $(M, M)$, the possible pairs are $(S, T)=(V_i, U_j)$ and $(S, T)=(V_i, Z_j)$ for $i,j \in \{1, \dotsc, r\}$.
	The right pad of~$S$ has length $P+r$ and the left pad of~$T$ has length at~most $P+r$, so~the inequality in~\eqref{eq:gap-overlap-bound} gives
	\[|\operatorname{ov}(S, T)| \le 2P+2r+2M+1 = \beta.\]

	Finally, suppose that $Q(S, T)=3$.
	Since $p_1, \dotsc, p_r$ are different, \Cref{tab:input-gap-sequences} shows that the five types of~pairs in~the statement and $(Z_i, U_i)$ are the only possibilities for distinct $S$ and $T$.
	For the latter pair, the inequality in~\eqref{eq:gap-overlap-bound} gives
	\[|\operatorname{ov}(Z_i, U_i)| \le P+2M+p_i+1 = 4P+17r+i+1 \le \beta.\]
	This proves the claim.
\end{proof}

\begin{table}[tbp]
	\centering
	\begin{tabular}{@{}ccc@{}}
		\toprule
		Sequence & Prefix & Suffix \\
		\midrule
		$(p_i, M)$ & $X_i$ & $X_i, U_i$ \\
		$(M, p_i)$ & $Y_i, V_i$ & $Y_i, Z_i$ \\
		$(M, M)$ & $U_i, Z_i$ & $V_i$ \\
		\midrule
		$(M, p_i, M)$ & $Y_i, V_i$ & $X_i, U_i$ \\
		$(p_i, M, p_i)$ & $X_i$ & $Y_i$ \\
		$(M, M, p_i)$ & $U_i, Z_i$ & $Z_i$ \\
		$(p_i, M, M)$ & --- & $V_i$ \\
		\bottomrule
	\end{tabular}
	\caption{The list of~strings in~$\mathcal{S}$ whose gap sequences begin or~end with the sequence in~the first column.
	In~the row for $(M, M)$, the index $i$ ranges from $1$ to~$r$ in~both columns.
	In~every other row, each string has the same index as~$p_i$.}
	\label{tab:input-gap-sequences}
\end{table}

\begin{table}[tbp]
	\centering
	\begin{tabular}{@{}ccc@{}}
		\toprule
		$(S, T)$ & $\operatorname{ov}(S, T)$ & $|\operatorname{ov}(S, T)|$ \\
		\midrule
		$(X_i, Y_i)$ & $[P \mid M\ p_i\ M \mid P]$ & $5P+17r+i+1$ \\
		$(X_i, V_i)$ & $[P \mid M\ p_i\ M \mid P]$ & $5P+17r+i+1$ \\
		$(U_i, Y_i)$ & $[P \mid M\ p_i\ M \mid P]$ & $5P+17r+i+1$ \\
		$(U_i, V_i)$ & $[P \mid M\ p_i\ M \mid P]$ & $5P+17r+i+1$ \\
		$(Y_i, X_i)$ & $[P+2r \mid p_i\ M\ p_i \mid P+4r]$ & $5P+16r+2i+1$ \\
		\bottomrule
	\end{tabular}
	\caption{The overlaps and their lengths for the pairs in~\Cref{lem:overlaps-above-beta}, where $1 \le i \le r$.}
	\label{tab:large-input-overlaps}
\end{table}

\Cref{tab:large-input-overlaps} shows the overlaps and their lengths for the pairs in~\Cref{lem:overlaps-above-beta}.
The first four pairs have overlaps of~length $\alpha_i$, and the last pair has an~overlap of~length $\alpha_i-r+i \le \alpha_i$.

We~now verify that every merge in~Steps~(i)--(iv) uses a~maximum overlap among the current strings.
For each step, we~compare the overlap of~the prescribed pair with all current overlaps.

\paragraph*{Step~(i)}
We~first prove the validity of~Step~(i) by~induction over $i = r, r-1, \dotsc, 1$.
Suppose that we~have completed the merges for indices $r, r-1, \dotsc, i+1$ using maximum overlaps.
The set of~current strings is
\[\{A_j, B_j, Z_j \mid i < j \le r\} \cup \{X_j, Y_j, U_j, V_j, Z_j \mid 1 \le j \le i\}.\]

We~show that the maximum length of~an~overlap among the current strings is~$\alpha_i$.
By~\Cref{lem:overlaps-above-beta}, overlaps between strings with different indices have length at~most $\beta$.
For each $j \le i$, \Cref{lem:overlaps-above-beta} and \Cref{tab:large-input-overlaps} show that the maximum length of~an~overlap between distinct strings with index $j$ is~$\alpha_j$.
For $j > i$, the current strings with index $j$ are $A_j = \langle X_j, Y_j\rangle$, $B_j = \langle U_j, V_j\rangle$, and $Z_j$.
Since no~pair of~distinct strings in~$\{ A_j, B_j, Z_j \}$ has $S$ as~a~suffix of~one string and $T$ as~a~prefix of~the other for any pair $(S,T)$ in~\Cref{tab:large-input-overlaps}, every overlap involving $A_j$, $B_j$, or~$Z_j$ with $j > i$ has length at~most $\beta \le \alpha_i$, by~\Cref{lem:overlaps-above-beta}.
Thus, the inequalities $\beta \le \alpha_i$ and $\alpha_j < \alpha_i$ for $j < i$ show that the maximum length is~$\alpha_i$.

We~can therefore merge $X_i$ with $Y_i$ to~obtain $A_i$ using an~overlap of~length $\alpha_i$.
The pair $(U_i, V_i)$ still has overlap of~length $\alpha_i$ after this merge, so~it~can be~selected next.
Hence, every merge in~Step~(i) uses a~maximum overlap by~induction.

\paragraph*{Step~(ii)}
After Step~(i), no~pair in~\Cref{tab:large-input-overlaps} occurs as~a~suffix and a~prefix of~two distinct current strings, so~all overlaps have length at~most $\beta$.
Each pair $(B_i, Z_i)$ in~Step~(ii) has overlap of~length $\beta$.
Since the maximum length of~an~overlap cannot increase, each of~these pairs has maximum overlap when it~is~selected.

\paragraph*{Step~(iii)}
At~the start of~Step~(iii), the current strings are $A_i = \langle X_i, Y_i\rangle$ and $C_i = \langle U_i, V_i, Z_i\rangle$ for $1 \le i \le r$.
By~\Cref{lem:gap-word-overlap}, we~have
\begin{align*}
	|\operatorname{ov}(A_i, A_j)| &= |\operatorname{ov}(Y_i, X_j)| = P+2r < \gamma &&\text{for } i \ne j, \\
	|\operatorname{ov}(C_i, C_j)| &= |\operatorname{ov}(Z_i, U_j)| = P+1 < \gamma &&\text{for } i \ne j, \\
	|\operatorname{ov}(A_i, C_j)| &= |\operatorname{ov}(Y_i, U_j)| = 2P + 4r + 1 = \gamma, \\
	|\operatorname{ov}(C_i, A_j)| &= |\operatorname{ov}(Z_i, X_j)| \le 2P+3r+i+1 \le \gamma.
\end{align*}
Thus, every overlap has length at~most $\gamma$.
Each pair $(A_i, C_i)$ still has overlap of~length $\gamma$ when it~is~selected, so~Step~(iii) is~valid.

\paragraph*{Step~(iv)}
At~the start of~Step~(iv), the current strings are $T_i = \langle X_i, Y_i, U_i, V_i, Z_i\rangle$ for $1 \le i \le r$.
For $i \ne j$, we~have $\operatorname{ov}(T_i, T_j) = \operatorname{ov}(Z_i, X_j)$.
We~show that this overlap is~empty.
Recall that
\[Z_i = [P+r \mid M\ M\ p_i \mid 0], \qquad X_j = [P+2r \mid p_j\ M\ p_j\ M \mid P].\]
$Z_i$ ends with the suffix ${\tt 1}{\tt 0}^{p_i-1}{\tt 1}$ and it~does not occur in~$X_j$ because $p_i$ differs from both $p_j$ and $M$.
Hence, any nonempty overlap must have the form ${\tt 0}^{a}{\tt 1}$ with $0 \le a \le p_i-1$.
However, $X_j$ begins with ${\tt 0}^{P+2r}{\tt 1}$, so~such an~overlap would require $a = P+2r$.
Since $a \le p_i-1 < P+2r$, no~such overlap exists.
Thus, $|\operatorname{ov}(T_i, T_j)| = |\operatorname{ov}(Z_i, X_j)| = 0$ for $i \ne j$.
Therefore, concatenating $T_1, \dotsc, T_r$ without overlap in~Step~(iv) is~a~valid execution of~the greedy algorithm.
This completes the proof that the procedure is~a~valid execution of~the greedy algorithm.

\subsection{A~Lower Bound for TGREEDY}
\label{sec:tgreedy-lower-bound}

The same construction also gives a~lower bound for TGREEDY.
Its first stage can extract the strings $A_i$ and $C_i$ as~representatives, and its second stage can reproduce Steps~(iii) and~(iv).

\begin{corollary}\label{cor:tgreedy-lower}
	For every $\varepsilon>0$, there~is a~set of~binary strings on~which TGREEDY may output a~common superstring of~length at~least $(12/5-\varepsilon)\operatorname{OPT}$.
\end{corollary}

Here, ``may'' means that TGREEDY produces this output for some way of~breaking ties in~both stages.

\begin{proof}
	Fix $r\ge1$ and set $P=r^2$ in~the construction above.
	We~first verify that the first stage can extract exactly the representatives $\{A_i,C_i\mid 1\le i\le r\}$.
	The gap representations give the following self-overlaps:
	\begin{align*}
		|\!\operatorname{ov}(X_i,X_i)| &= 4P+11r+i+1,\\
		|\!\operatorname{ov}(Y_i,Y_i)| &= 4P+13r+i+1,\\
		|\!\operatorname{ov}(U_i,U_i)| &= 3P+8r+1,\\
		|\!\operatorname{ov}(V_i,V_i)| &= 3P+9r+1,\\
		|\!\operatorname{ov}(Z_i,Z_i)| &= P+r+1,\\
		|\!\operatorname{ov}(A_i,A_i)| &= \delta_i=5P+16r+2i+1,\\
		|\!\operatorname{ov}(B_i,B_i)| &= 4P+17r+1,\\
		|\!\operatorname{ov}(C_i,C_i)| &= 4P+17r+i+1.
	\end{align*}
	In~particular, every input string and every $B_i$ has self-overlap smaller than $\beta$.
	For the strings $A_i$, we~have
	\[\delta_i=\alpha_i-r+i\le\alpha_i,\qquad \delta_i\ge\delta_1=\beta+(r-1)^2+1>\beta.\]

	Before each merge of~Step~(i), extract every current $A_j$ whose self-overlap exceeds that merge's overlap, in~decreasing order of~self-overlap.
	Deleting such strings cannot increase any remaining overlap or~remove either endpoint of~a~later prescribed merge.
	The proof of~Step~(i) therefore still bounds every overlap between distinct current strings by~the next prescribed value $\alpha_i$.
	After these extractions, every remaining $A_j$ has self-overlap at~most $\alpha_i$, while every other current string has self-overlap below $\beta\le\alpha_i$.
	Thus, each prescribed merge is~valid.
	After these merges, all overlaps between distinct current strings are at~most $\beta$, as~in~Step~(ii).
	Extract any remaining $A_i$ in~decreasing order of~$\delta_i$; these extractions are valid because $\delta_i>\beta$.
	The current strings are now precisely $B_i$ and $Z_i$ for $1\le i\le r$.

	Next perform the merges $B_i\to Z_i$ of~Step~(ii), each with overlap $\beta$.
	The new string $C_i$ has self-overlap at~most $\beta$, so~these merges remain valid when ties between an~extraction and a~merge are broken in~favor of~the merge.
	The remaining strings are the $C_i$, whose mutual overlaps have length $P+1$ for distinct indices.
	Since every $C_i$ has self-overlap larger than $P+1$, extract them in~decreasing order of~self-overlap.
	This completes the first stage with exactly the representatives $\{A_i,C_i\mid 1\le i\le r\}$.

	These representatives are substring-free.
	Different indices have distinct private gaps, and the gap sequence $(p_i,M,p_i,M,p_i)$ of~$A_i$ is~not a~consecutive segment of~$(M,M,p_i,M,M,p_i)$, the gap sequence of~$C_i$.
	Since $C_i$ has more ones than $A_i$, the reverse containment is~also impossible.
	Thus, the greedy algorithm in~the second stage can follow Steps~(iii) and~(iv), merging each $A_i$ with $C_i$ and concatenating the resulting strings $T_i$.
	Its output $g$ has the same length as~before, and the same common superstring $H$ gives
	\[\frac{|g|}{\operatorname{OPT}}\ge\frac{|g|}{|H|}=\frac{12Pr+(115r^2+7r)/2}{5Pr+27r^2+4P+19r+1}.\]
	With $P=r^2$, this lower bound tends to~$12/5$ as~$r\to\infty$, proving the corollary.
\end{proof}

 \section{Tight Bound for~Strings of~Length Three}
\label{sec:length-three}

In~this section, we~consider inputs in~which all strings have the same length~$k$.
Let $\rho_k$ be~the approximation ratio of~the greedy algorithm on~such inputs, that~is, the supremum of~$|g|/\operatorname{OPT}$ over all finite alphabets, all nonempty sets~$\mathcal{S}$ of~strings of~length~$k$, and all outputs~$g$ of~the greedy algorithm on~$\mathcal{S}$.
Cazaux and Rivals~\cite{CR18} proved that $5/3\le\rho_3\le2$.
We~show that $\rho_3 = 9/5$, starting with the upper bound.

\begin{theorem}[Formal version of~\Cref{thm:length-three-informal}, upper bound]\label{thm:length-three}
	For every nonempty set~$\mathcal{S}$ of~strings of~length three, every output~$g$ of~the greedy algorithm on~$\mathcal{S}$ satisfies
	\[|g|<\frac{9}{5}\operatorname{OPT}.\]
	Consequently, $\rho_3\le9/5$.
\end{theorem}

Then, we~show that this bound is~tight.

\begin{proposition}[Formal version of~\Cref{thm:length-three-informal}, lower bound]\label{prop:length-three-tight}
	For every integer $\ell\ge1$, there is~a~set of~$5\ell$ strings of~length three that has a~common superstring of~length $5\ell+2$, whereas the greedy algorithm may output a~common superstring of~length $9\ell$.
	Consequently, $\rho_3\ge9/5$.
\end{proposition}
Here, ``may'' means that the greedy algorithm produces this output for some way of~breaking ties.
Unlike the bound of~\Cref{thm:greedy-lower-bound}, this bound does not extend to~every tie-breaking rule by~\cite{Nikolaev21} as the transformation used there is not length-preserving.

Throughout this section, $\mathcal{S}$ is~a~set of~$n$ strings of~length three, and $g$ is~the output of~an~arbitrary run of~the greedy algorithm on~$\mathcal{S}$.
By~\Cref{lem:greedy-ends}, the greedy algorithm first merges strings with overlap two, which we~call the overlap-two stage, then strings with overlap one, which we~call the overlap-one stage, and then concatenates the remaining strings.
Let $\kappa_2$ and~$\kappa_1$ be~the numbers of~strings after the overlap-two and the overlap-one stage, respectively, so~that $\kappa_1\le\kappa_2$.
There are $n-\kappa_2$ merges with overlap two and $\kappa_2-\kappa_1$ merges with overlap one, so
\[|g|=3n-2(n-\kappa_2)-(\kappa_2-\kappa_1)=n+\kappa_1+\kappa_2.\]

We~describe the overlap-two stage in~terms of~the directed graph~$G_2$ whose vertices are the strings of~length two that are prefixes or~suffixes of~strings in~$\mathcal{S}$, and in~which every string $abc\in\mathcal{S}$ is~an~edge from $ab$ to~$bc$; it~is~a~subgraph of~the de~Bruijn graph~\cite{DeBrujin46}.
Thus, $G_2$ has $n$ edges, and two input strings overlap by~two symbols exactly when the first one ends at~the vertex where the second one starts, so~the overlap-two stage merges them along trails of~$G_2$.
For the overlap-one stage, we~create a~multigraph~$G_1$ on~symbols, where every string left by~the overlap-two stage corresponds to~an~edge from its first to~its last symbol.
Thus, $G_1$ is~built like~$G_2$, but for the strings left by~the overlap-two stage and with their prefixes and suffixes of~length one.

\begin{graybox}
	As~a~running example, consider the $n=13$ strings
	\[\mathcal{S}=\{\texttt{aaa},\texttt{aac},\texttt{aba},\texttt{aca},\texttt{baa},\texttt{bab},\texttt{cac},\texttt{dba},\texttt{efg},\texttt{fge},\texttt{geg},\texttt{ege},\texttt{gef}\}.\]
	The graph~$G_2$ is~shown below at~the top, with every edge labeled by~its string.
	The greedy algorithm may first merge \texttt{dba}, \texttt{baa}, and~\texttt{aac} into \texttt{dbaac}, \texttt{aba} and~\texttt{bab} into \texttt{abab}, and \texttt{cac} and~\texttt{aca} into \texttt{caca}, as~well as~\texttt{fge}, \texttt{gef}, and~\texttt{efg} into \texttt{fgefg} and \texttt{ege} and~\texttt{geg} into \texttt{egeg}.
	Then no~two of~the current strings \texttt{dbaac}, \texttt{abab}, \texttt{caca}, \texttt{aaa}, \texttt{fgefg}, and~\texttt{egeg} overlap by~two symbols, so~$\kappa_2=6$.
	The edges of~$G_1$, shown at~the bottom, correspond to~the same six strings and are drawn in~the same colors.
	In~the overlap-one stage, the greedy algorithm may merge \texttt{caca} and~\texttt{abab} into \texttt{cacabab}, and then \texttt{dbaac} and~\texttt{cacabab} into \texttt{dbaacacabab}.
	After that, no~two current strings overlap, so~$\kappa_1=4$, and the greedy algorithm outputs $g=\texttt{dbaacacababaaafgefgegeg}$ of~length $23=n+\kappa_1+\kappa_2$.
	\begin{center}
		\begin{tikzpicture}[tq/.style={draw=blue!30!gray, line width=1pt}, tl/.style={draw=violet!30!gray, line width=1pt}, tb/.style={draw=orange!40!gray, line width=1pt}, tc/.style={draw=green!25!gray, line width=1pt}, tr/.style={draw=teal!45!gray, line width=1pt}, tz/.style={draw=magenta!30!gray, line width=1pt}]
			\foreach \s/\x/\y in {db/0/0, ba/1.8/0, aa/3.6/0, ac/5.4/0, ca/7.2/0, ab/1.8/1.6, eg/9/0, ge/10.8/0, ef/11.7/1.4, fg/12.6/0}
				\node[v] (\s) at (\x, \y) {\strut\ttfamily\s};
			\path (db) edge[le, tq] node[below] {\scriptsize\ttfamily dba} (ba);
			\path (ba) edge[le, tq] node[below] {\scriptsize\ttfamily baa} (aa);
			\path (aa) edge[le, tq] node[below] {\scriptsize\ttfamily aac} (ac);
			\path (aa) edge[le, tl, loop above, out=125, in=55, min distance=7mm] node[above] {\scriptsize\ttfamily aaa} (aa);
			\path (ab) edge[le, tb, bend left=30] node[right] {\scriptsize\ttfamily aba} (ba);
			\path (ba) edge[le, tb, bend left=30] node[left] {\scriptsize\ttfamily bab} (ab);
			\path (ac) edge[le, tc, bend left=35] node[above] {\scriptsize\ttfamily aca} (ca);
			\path (ca) edge[le, tc, bend left=35] node[below] {\scriptsize\ttfamily cac} (ac);
			\path (eg) edge[le, tz, bend left=35] node[above] {\scriptsize\ttfamily ege} (ge);
			\path (ge) edge[le, tz, bend left=35] node[below] {\scriptsize\ttfamily geg} (eg);
			\path (ge) edge[le, tr] node[above left] {\scriptsize\ttfamily gef} (ef);
			\path (ef) edge[le, tr] node[above right] {\scriptsize\ttfamily efg} (fg);
			\path (fg) edge[le, tr] node[below] {\scriptsize\ttfamily fge} (ge);
			\foreach \s/\x in {d/1.2, c/2.8, a/4.4, b/6, e/9.4, g/11, f/12.6}
				\node[v] (\s) at (\x, -2.1) {\strut\ttfamily\s};
			\path (d) edge[le, tq] node[below] {\scriptsize\ttfamily dbaac} (c);
			\path (c) edge[le, tc] node[below] {\scriptsize\ttfamily caca} (a);
			\path (a) edge[le, tb] node[below] {\scriptsize\ttfamily abab} (b);
			\path (a) edge[le, tl, loop above, out=125, in=55, min distance=7mm] node[above] {\scriptsize\ttfamily aaa} (a);
			\path (e) edge[le, tz] node[below] {\scriptsize\ttfamily egeg} (g);
			\path (f) edge[le, tr] node[below] {\scriptsize\ttfamily fgefg} (g);
		\end{tikzpicture}
	\end{center}
\end{graybox}

Before giving the details, we~explain the idea of~the proof of~\Cref{thm:length-three}.
The overlap-two stage leaves $\kappa_2$ strings, which correspond to~$\kappa_2$ trails of~$G_2$ that cover all its edges.
We~say that a~trail is~\emph{closed} if~it~ends at~the vertex where it~starts, and \emph{open} otherwise.
Let $\pi(G_2)$ be~the minimum number of~trails that cover all edges of~$G_2$, then \Cref{lem:opt-length-three} shows that $\operatorname{OPT}\ge n+\pi(G_2)+1$.
No~trail formed by~the overlap-two stage ends at~the vertex where another one starts, since otherwise their strings would overlap by~two, so~at~most $\pi(G_2)$ of~these trails are open by~\Cref{lem:first-stage}.
It~remains to~bound the number of~closed trails.
If~all closed trails had at~least three edges, then the $n$ edges of~$G_2$ would lie on~$\kappa_2$ trails, at~most $\pi(G_2)$ of~which are open and have at~least one edge each, while the others have at~least three edges each, so~$n\ge3\kappa_2-2\pi(G_2)$, and using the trivial bound $\kappa_1\le\kappa_2$ would give $|g|\le n+2\kappa_2\le\frac{5n+4\pi(G_2)}{3}<\frac53\operatorname{OPT}$, which is~the value $\rho_3=2-1/3$ conjectured by~Cazaux and Rivals~\cite{CR18}.
Hence, the main difficulty lies in~closed trails with one or~two edges, formed by~a~string $aaa$ or~by~two strings $aba$ and~$bab$ for distinct symbols $a$ and~$b$.
Such trails make $\kappa_2$ large, and our main lemma, \Cref{lem:length-three-main}, shows that when $\kappa_{2}$ is~large, the edges of~$G_1$ can be~covered by~few trails, so~the overlap-one stage merges many strings and $\kappa_1$ is~small.

The following example shows that short closed trails make \Cref{thm:length-three} tight: the overlap-two stage may form only closed trails with two and three edges, after which no~two current strings overlap.

\begin{proof}[Proof of~\Cref{prop:length-three-tight}]
	Let $\texttt{a},\texttt{b},\texttt{x}_1,\dotsc,\texttt{x}_\ell$ be~distinct symbols, and let $\mathcal{S}$ be~the set of~the length-three substrings of
	\[h=\texttt{ab}\texttt{x}_1\texttt{a}\texttt{x}_1\texttt{ab}\dotsm\texttt{x}_\ell\texttt{a}\texttt{x}_\ell\texttt{ab}.\]
	It~consists of~the $5\ell$ distinct strings $\texttt{ab}\texttt{x}_i$, $\texttt{b}\texttt{x}_i\texttt{a}$, $\texttt{x}_i\texttt{a}\texttt{x}_i$, $\texttt{a}\texttt{x}_i\texttt{a}$, and~$\texttt{x}_i\texttt{ab}$ for $i\in[\ell]$, and $|h|=5\ell+2$.
	For every~$i$, the edges $\texttt{a}\texttt{x}_i\texttt{a}$ and~$\texttt{x}_i\texttt{a}\texttt{x}_i$ of~$G_2$ form a~closed trail, and so~do the edges $\texttt{b}\texttt{x}_i\texttt{a}$, $\texttt{x}_i\texttt{ab}$, and~$\texttt{ab}\texttt{x}_i$.
	Consecutive edges of~a~trail overlap by~two symbols, so~the overlap-two stage may merge the strings along each of~these $2\ell$ trails, which gives the current strings $\texttt{a}\texttt{x}_i\texttt{a}\texttt{x}_i$ and~$\texttt{b}\texttt{x}_i\texttt{ab}\texttt{x}_i$.
	No~two of~these strings overlap.
	Indeed, each of~them starts and ends with the same string, $\texttt{a}\texttt{x}_i$ or~$\texttt{b}\texttt{x}_i$, which is~different for different strings, and contains~$\texttt{x}_i$ only at~its second and last positions.
	Hence, if~a~nonempty proper suffix of~one of~them is~a~proper prefix of~another one, then it~ends with some~$\texttt{x}_i$, so~it~consists of~the first two symbols of~the latter string, and the former string ends with these two symbols, which is~impossible.
	Therefore, $\kappa_1=\kappa_2=2\ell$, and $|g|=n+\kappa_1+\kappa_2=9\ell$.
	Since $\operatorname{OPT}\le|h|$ and $9\ell/(5\ell+2)$ tends to~$9/5$ as~$\ell\to\infty$, we~obtain $\rho_3\ge9/5$, and \Cref{thm:length-three} gives $\rho_3\le9/5$.
\end{proof}

\subsection{\texorpdfstring{A~Trail Cover of~$G_2$}{A Trail Cover of G2}}
\label{sec:length-three-trails}

A~\emph{trail}~$P$ in~a~directed multigraph is~a~walk without repeated edges.
A~trail $P$ with $m\ge1$ edges visits vertices $P_0,P_1,\dotsc,P_m$.
We~call every index $i$ with $0\le i\le m$ a~position of~$P$; it~is~\emph{internal} if~$0<i<m$.
The trail $P$ is~\emph{closed} if~$P_0=P_m$ and \emph{open} otherwise.
A~\emph{trail cover} of~a~directed multigraph~$H$ is~a~set of~edge-disjoint trails that contain all edges of~$H$, and $\pi(H)$ denotes the minimum size of~a~trail cover of~$H$.
The \emph{weak components} of~$H$ are the connected components of~the undirected multigraph obtained by~ignoring the directions of~the edges, and $H$ is~weakly connected if~it~has one weak component.
For a~vertex~$v$ of~$H$, let $\deg^+(v)$ and~$\deg^-(v)$ be~its out-degree and in-degree.
The \emph{excess} of~$H$ is~$\operatorname{exc}(H)=\sum_v\max\{\deg^+(v)-\deg^-(v),0\}$, the total number of~outgoing edges in~excess of~incoming ones.

\begin{lemma}[{\cite[Lemma~2]{GMS80}}]\label{lem:trail-number}
	If~$H$ is~weakly connected and has at~least one edge, then $\pi(H)=\max\{\operatorname{exc}(H),1\}$.
\end{lemma}

We~can view the overlap-two stage of~the greedy algorithm as~building a~trail cover of~$G_2$, in~which every current string is~the overlap merge of~the edges of~a~trail in~their order on~the trail.
Let $\mathcal{T}$ be~the resulting trail cover, which consists of~$\kappa_2$ trails.
Consecutive vertices of~a~trail~$P$ of~$G_2$ with $m$ edges overlap in~one symbol, so~there are symbols $a_0,\dotsc,a_{m+1}$, which we~call the \emph{symbols of~$P$}, such that $P_i=a_ia_{i+1}$ for $0\le i\le m$.
Note that no~trail of~$\mathcal{T}$ ends at~the first vertex of~another trail of~$\mathcal{T}$, since otherwise their current strings would overlap by~two after the overlap-two stage.
The next lemma shows that $\mathcal{T}$ has the minimum possible number of~open trails.

\begin{lemma}\label{lem:first-stage}
	Let $\nu$ be~the number of~weak components of~$G_2$ that contain no~open trail of~$\mathcal{T}$.
	Then $\mathcal{T}$ has exactly $\pi(G_2)-\nu$ open trails.
\end{lemma}
\begin{proof}
	Since no~trail of~$\mathcal{T}$ ends at~the first vertex of~another one, for every vertex~$v$ of~$G_2$, either no~trail of~$\mathcal{T}$ ends at~$v$, or~no~trail of~$\mathcal{T}$ starts at~$v$, or~the only trail of~$\mathcal{T}$ that starts or~ends at~$v$ is~a~closed trail that starts at~$v$.
	The number of~trails of~$\mathcal{T}$ that start at~$v$ exceeds the number of~those that end at~$v$ by~$\deg^+(v)-\deg^-(v)$, so~in~each case exactly $\max\{\deg^+(v)-\deg^-(v),0\}$ open trails of~$\mathcal{T}$ start at~$v$.
	Thus, every weak component~$H$ of~$G_2$ contains exactly $\operatorname{exc}(H)$ open trails of~$\mathcal{T}$, and summing $\pi(H)=\max\{\operatorname{exc}(H),1\}$ over all weak components~$H$ of~$G_2$ proves the lemma.
\end{proof}

\begin{graybox}
	Recall the graph~$G_2$ of~the running example, in~which the trails of~$\mathcal{T}$ are drawn in~different colors.
	\begin{center}
		\begin{tikzpicture}[tq/.style={draw=blue!30!gray, line width=1pt}, tl/.style={draw=violet!30!gray, line width=1pt}, tb/.style={draw=orange!40!gray, line width=1pt}, tc/.style={draw=green!25!gray, line width=1pt}, tr/.style={draw=teal!45!gray, line width=1pt}, tz/.style={draw=magenta!30!gray, line width=1pt}]
			\foreach \s/\x/\y in {db/0/0, ba/1.8/0, aa/3.6/0, ac/5.4/0, ca/7.2/0, ab/1.8/1.6, eg/9/0, ge/10.8/0, ef/11.7/1.4, fg/12.6/0}
				\node[v] (\s) at (\x, \y) {\strut\ttfamily\s};
			\path (db) edge[le, tq] node[below] {\scriptsize\ttfamily dba} (ba);
			\path (ba) edge[le, tq] node[below] {\scriptsize\ttfamily baa} (aa);
			\path (aa) edge[le, tq] node[below] {\scriptsize\ttfamily aac} (ac);
			\path (aa) edge[le, tl, loop above, out=125, in=55, min distance=7mm] node[above] {\scriptsize\ttfamily aaa} (aa);
			\path (ab) edge[le, tb, bend left=30] node[right] {\scriptsize\ttfamily aba} (ba);
			\path (ba) edge[le, tb, bend left=30] node[left] {\scriptsize\ttfamily bab} (ab);
			\path (ac) edge[le, tc, bend left=35] node[above] {\scriptsize\ttfamily aca} (ca);
			\path (ca) edge[le, tc, bend left=35] node[below] {\scriptsize\ttfamily cac} (ac);
			\path (eg) edge[le, tz, bend left=35] node[above] {\scriptsize\ttfamily ege} (ge);
			\path (ge) edge[le, tz, bend left=35] node[below] {\scriptsize\ttfamily geg} (eg);
			\path (ge) edge[le, tr] node[above left] {\scriptsize\ttfamily gef} (ef);
			\path (ef) edge[le, tr] node[above right] {\scriptsize\ttfamily efg} (fg);
			\path (fg) edge[le, tr] node[below] {\scriptsize\ttfamily fge} (ge);
		\end{tikzpicture}
	\end{center}
	The graph~$G_2$ has two weak components.
	In~the first one, \texttt{db} is~the only vertex with more outgoing than incoming edges, so~its excess is~$1$; in~the second one, every vertex has as~many outgoing as~incoming edges, so~its excess is~$0$, and $\pi(G_2)=1+1=2$ by~\Cref{lem:trail-number}.
	The first component contains the open trail~$\texttt{db}\to\texttt{ba}\to\texttt{aa}\to\texttt{ac}$, whereas the second one consists only of~closed trails, so~$\nu=1$, and $\mathcal{T}$ indeed has exactly $\pi(G_2)-\nu=1$ open trails.
\end{graybox}

We~now bound $\operatorname{OPT}$ from below.

\begin{lemma}\label{lem:opt-length-three}
	\[\operatorname{OPT}\ge n+\pi(G_2)+1.\]
\end{lemma}
\begin{proof}
	Let $\operatorname{OPT}=|\langle t_1,\dotsc,t_n\rangle|$ for an~ordering $t_1,\dotsc,t_n$ of~$\mathcal{S}$, and split this ordering into maximal blocks in~which consecutive strings overlap by~two symbols.
	The strings of~every block are the consecutive edges of~a~trail of~$G_2$, and every input string lies in~exactly one block, so~these trails form a~trail cover of~$G_2$.
	Hence, the number~$r$ of~blocks is~at~least $\pi(G_2)$.
	Since $n-r$ pairs of~consecutive strings overlap by~two symbols and the other $r-1$ pairs by~at~most one symbol,
	\[\operatorname{OPT}\ge3n-2(n-r)-(r-1)=n+r+1\ge n+\pi(G_2)+1.\qedhere\]
\end{proof}

Let $G_1$ be~the directed multigraph whose vertices are symbols and which has an~edge~$e_P$ from $a_0$ to~$a_{m+1}$ for every trail $P\in\mathcal{T}$ with symbols $a_0,\dotsc,a_{m+1}$.
In~other words, $e_P$ goes from the first to~the last symbol of~the current string of~$P$.
Now, we~analyze the overlap-one stage.

\begin{lemma}\label{lem:second-stage}
	\[2\kappa_1\le\kappa_2+\pi(G_1).\]
\end{lemma}
\begin{proof}
	Let $\mathcal{R}$ be~the set of~the $\kappa_2$ current strings after the overlap-two stage.
	The set $\mathcal{R}$ is~substring-free, and the rest of~the run of~the greedy algorithm is~a~run of~the greedy algorithm on~$\mathcal{R}$ with output~$g$.
	This run performs $\kappa_2-\kappa_1$ merges with overlap one, so~the compression of~$g$, that~is, the total length of~the strings in~$\mathcal{R}$ minus~$|g|$, is~$\kappa_2-\kappa_1$.
	Every string of~$\mathcal{R}$ is~the current string of~a~unique trail $P\in\mathcal{T}$ and thus corresponds to~the edge~$e_P$ of~$G_1$.
	Take a~minimum trail cover of~$G_1$ and order~$\mathcal{R}$ by~listing, trail after trail, the strings that correspond to~the edges of~each trail in~their order on~the trail.
	If~two edges are consecutive on~a~trail of~$G_1$, then the first ends at~the symbol at~which the second starts, so~the corresponding strings overlap by~at~least one symbol.
	Hence, all pairs of~consecutive strings in~this order except at~most $\pi(G_1)-1$ overlap by~at~least one symbol, so~the overlap merge of~$\mathcal{R}$ in~this order has compression at~least $\kappa_2-\pi(G_1)$.
	By~\Cref{thm:half-compression} applied to~the input~$\mathcal{R}$, $\kappa_2-\kappa_1\ge(\kappa_2-\pi(G_1))/2$.
\end{proof}

Recall that $|g|=n+\kappa_1+\kappa_2$ and that $\operatorname{OPT}\ge n+\pi(G_2)+1$ by~\Cref{lem:opt-length-three}.
Since $\frac95n+\frac85\pi(G_2)<\frac95\bigl(n+\pi(G_2)+1\bigr)$, it~suffices to~show that $|g|\le\frac95n+\frac85\pi(G_2)$, that~is,
\begin{equation}\label{eq:length-three-goal}
	\kappa_1+\kappa_2\le\frac45n+\frac85\pi(G_2).
\end{equation}
We~have two upper bounds on~$\kappa_1$: the trivial bound $\kappa_1\le\kappa_2$ and the bound $\kappa_1\le\frac12\bigl(\kappa_2+\pi(G_1)\bigr)$ of~\Cref{lem:second-stage}.
Adding one fifth of~the first bound to~four fifths of~the second one, we~get
\[\kappa_1+\kappa_2\le\frac15\kappa_2+\frac25\bigl(\kappa_2+\pi(G_1)\bigr)+\kappa_2=\frac85\kappa_2+\frac25\pi(G_1),\]
so~\eqref{eq:length-three-goal} holds whenever
\begin{equation}\label{eq:length-three-needed}
	\pi(G_1)\le2n+4\pi(G_2)-4\kappa_2.
\end{equation}

To~understand the right-hand side of~\eqref{eq:length-three-needed}, consider a~trail $P\in\mathcal{T}$ with $m_P$ edges.
Its edges are $m_P$ input strings, so~they add $2m_P$ to~the right-hand side through~$n$.
The trail itself is~one of~the $\kappa_2$ current strings after the overlap-two stage, so~it~subtracts four through~$\kappa_2$.
If~$P$ is~open, then it~also adds four through~$\pi(G_2)$, since $\pi(G_2)$ is~the number of~open trails of~$\mathcal{T}$ plus~$\nu$ by~\Cref{lem:first-stage}.
Accordingly, we~define the \emph{contribution} of~$P$ as~$\operatorname{contr}(P)=2m_P-4$ if~$P$ is~closed and $\operatorname{contr}(P)=2m_P$ if~$P$ is~open, so~that
\[2n+4\pi(G_2)-4\kappa_2=4\nu+\sum_{P\in\mathcal{T}}\operatorname{contr}(P),\]
and \eqref{eq:length-three-needed} becomes
\[\pi(G_1)\le4\nu+\sum_{P\in\mathcal{T}}\operatorname{contr}(P).\]
Thus, \Cref{thm:length-three} follows from the following main lemma, which we~prove in~\Cref{sec:length-three-charging}.

\begin{lemma}\label{lem:length-three-main}
	\[\pi(G_1)\le4\nu+\sum_{P\in\mathcal{T}}\operatorname{contr}(P).\]
\end{lemma}

\begin{graybox}
	Recall the graphs $G_2$ and~$G_1$ of~the running example.
	\begin{center}
		\begin{tikzpicture}[tq/.style={draw=blue!30!gray, line width=1pt}, tl/.style={draw=violet!30!gray, line width=1pt}, tb/.style={draw=orange!40!gray, line width=1pt}, tc/.style={draw=green!25!gray, line width=1pt}, tr/.style={draw=teal!45!gray, line width=1pt}, tz/.style={draw=magenta!30!gray, line width=1pt}]
			\foreach \s/\x/\y in {db/0/0, ba/1.8/0, aa/3.6/0, ac/5.4/0, ca/7.2/0, ab/1.8/1.6, eg/9/0, ge/10.8/0, ef/11.7/1.4, fg/12.6/0}
				\node[v] (\s) at (\x, \y) {\strut\ttfamily\s};
			\path (db) edge[le, tq] node[below] {\scriptsize\ttfamily dba} (ba);
			\path (ba) edge[le, tq] node[below] {\scriptsize\ttfamily baa} (aa);
			\path (aa) edge[le, tq] node[below] {\scriptsize\ttfamily aac} (ac);
			\path (aa) edge[le, tl, loop above, out=125, in=55, min distance=7mm] node[above] {\scriptsize\ttfamily aaa} (aa);
			\path (ab) edge[le, tb, bend left=30] node[right] {\scriptsize\ttfamily aba} (ba);
			\path (ba) edge[le, tb, bend left=30] node[left] {\scriptsize\ttfamily bab} (ab);
			\path (ac) edge[le, tc, bend left=35] node[above] {\scriptsize\ttfamily aca} (ca);
			\path (ca) edge[le, tc, bend left=35] node[below] {\scriptsize\ttfamily cac} (ac);
			\path (eg) edge[le, tz, bend left=35] node[above] {\scriptsize\ttfamily ege} (ge);
			\path (ge) edge[le, tz, bend left=35] node[below] {\scriptsize\ttfamily geg} (eg);
			\path (ge) edge[le, tr] node[above left] {\scriptsize\ttfamily gef} (ef);
			\path (ef) edge[le, tr] node[above right] {\scriptsize\ttfamily efg} (fg);
			\path (fg) edge[le, tr] node[below] {\scriptsize\ttfamily fge} (ge);
			\foreach \s/\x in {d/1.2, c/2.8, a/4.4, b/6, e/9.4, g/11, f/12.6}
				\node[v] (\s) at (\x, -2.1) {\strut\ttfamily\s};
			\path (d) edge[le, tq] node[below] {\scriptsize\ttfamily dbaac} (c);
			\path (c) edge[le, tc] node[below] {\scriptsize\ttfamily caca} (a);
			\path (a) edge[le, tb] node[below] {\scriptsize\ttfamily abab} (b);
			\path (a) edge[le, tl, loop above, out=125, in=55, min distance=7mm] node[above] {\scriptsize\ttfamily aaa} (a);
			\path (e) edge[le, tz] node[below] {\scriptsize\ttfamily egeg} (g);
			\path (f) edge[le, tr] node[below] {\scriptsize\ttfamily fgefg} (g);
		\end{tikzpicture}
	\end{center}
	Here, $\mathcal{T}$ consists of~six trails: in~the first component, the open trail $\texttt{db}\to\texttt{ba}\to\texttt{aa}\to\texttt{ac}$ with three edges, the loop \texttt{aaa}, and the closed trails $\texttt{ab}\to\texttt{ba}\to\texttt{ab}$ and $\texttt{ca}\to\texttt{ac}\to\texttt{ca}$ with two edges each, and in~the second component, the closed trails $\texttt{eg}\to\texttt{ge}\to\texttt{eg}$ and $\texttt{fg}\to\texttt{ge}\to\texttt{ef}\to\texttt{fg}$ with two and three edges.
	Their contributions are $6$, $-2$, $0$, $0$, $0$, and~$2$, respectively.
	Since $\nu=1$, the right-hand side of~\Cref{lem:length-three-main} is~$4+6=10$.
	In~fact, one can cover $G_1$ with three trails, so~$\pi(G_1)=3$.
\end{graybox}

\subsection{\texorpdfstring{A~Trail Cover of~$G_1$}{A Trail Cover of G1}}
\label{sec:length-three-charging}

In~this subsection, we~prove \Cref{lem:length-three-main}, that~is, we~construct a~trail cover of~$G_1$ with at~most $4\nu+\sum_{P\in\mathcal{T}}\operatorname{contr}(P)$ trails.
Recall that every trail $P\in\mathcal{T}$ corresponds to~a~single edge~$e_P$ of~$G_1$.
The contribution~$\operatorname{contr}(P)$ is~positive for open trails and for closed trails with at~least three edges, but it~is~$0$ for closed trails with two edges and $-2$ for loops.
We~split the trails of~$\mathcal{T}$, and thus the edges of~$G_1$, into groups such that, except for at~most $\nu$ groups, every group has positive total contribution and its edges in~$G_1$ can be~covered by~at~most this many trails.
We~form the groups as~follows.
First, a~loop $aaa$ that shares its weak component of~$G_2$ with other edges lies on~another trail of~$\mathcal{T}$ that passes through~$aa$, and we~put the loop into the group of~this trail, which absorbs the loop; the remaining contribution $\operatorname{rem}(P)=\operatorname{contr}(P)-2\operatorname{loops}(P)$ of~a~trail~$P$, where $\operatorname{loops}(P)$ is~the number of~loops that $P$ absorbs, is~still nonnegative (\Cref{lem:shares}).
Second, a~closed trail with $\operatorname{rem}(P)=0$ uses only two symbols $a$ and~$b$, and all such trails on~$a$ and~$b$ together with their loops form a~\emph{bundle}, whose edges in~$G_1$ form a~single trail of~$G_1$ (\Cref{lem:bundle-trails,lem:bundles}); the total contribution of~a~bundle is~$0$, so~we~put it~into the group of~a~trail with positive remaining contribution whenever possible.
Third, the other trails, that~is, the open trails and the closed trails with $\operatorname{rem}(P)>0$, are called \emph{connectors}; we~attach every bundle that shares its weak component of~$G_2$ with a~connector to~one such connector, and the group of~a~connector~$Q$, formed by~$Q$, its loops, and the attached bundles, has total contribution~$\operatorname{rem}(Q)$ and can be~covered by~$\operatorname{rem}(Q)$ trails of~$G_1$ (\Cref{lem:local-cover}).
Finally, a~weak component of~$G_2$ without connectors is~a~single loop or~a~single bundle; it~has no~open trails, so~there are at~most $\nu$ such components, and we~cover each of~them by~one trail.

A~loop is~a~trail of~$\mathcal{T}$ that consists of~a~single edge from $aa$ to~$aa$, that~is, of~the string $aaa$ for a~symbol~$a$; its edge in~$G_1$ is~a~loop at~$a$.
Consider a~loop $aaa$ whose weak component in~$G_2$ has another edge.
Then another trail $P\in\mathcal{T}$ visits~$aa$, and since no~trail of~$\mathcal{T}$ ends at~the first vertex of~another one, it~neither starts nor ends at~$aa$, so~$P_i=aa$ for an~internal position~$i$ of~$P$.
Several trails, or~several positions of~one trail, may visit~$aa$; we~fix one such trail~$P$ and position~$i$ and say that $P$ \emph{absorbs} the loop $aaa$ at~position~$i$, so~every loop is~absorbed by~at~most one trail.
Note that $a_{i-1}\ne a\ne a_{i+2}$, since otherwise $P$ would contain the edge $aaa$.
For a~trail $P\in\mathcal{T}$ that is~not a~loop, let $\operatorname{loops}(P)$ be~the number of~loops that $P$ absorbs, and let $\operatorname{rem}(P)=\operatorname{contr}(P)-2\operatorname{loops}(P)$ be~the \emph{remaining contribution} of~$P$.

By~the following lemma, to~prove \Cref{lem:length-three-main}, it~suffices to~construct a~trail cover of~$G_1$ of~size at~most $\sum_P\operatorname{rem}(P)+\nu$, where the sum is~over the trails $P\in\mathcal{T}$ that are not loops.

\begin{lemma}\label{lem:shares}
	\[\sum_P\operatorname{rem}(P)+\nu\le4\nu+\sum_{P\in\mathcal{T}}\operatorname{contr}(P),\]
	where the sum is~over the trails $P\in\mathcal{T}$ that are not loops.
\end{lemma}
\begin{proof}
	By~the definition of~$\operatorname{rem}(P)$, the sum $\sum_P\operatorname{rem}(P)$ equals the sum of~the contributions of~the trails of~$\mathcal{T}$ that are not loops and of~the loops that are absorbed.
	Every loop that is~not absorbed forms a~weak component of~$G_2$ without open trails, so~these loops contribute at~least $-2\nu$, and therefore $\sum_P\operatorname{rem}(P)\le2\nu+\sum_{P\in\mathcal{T}}\operatorname{contr}(P)$, which implies the lemma.
\end{proof}

A~closed trail~$P$ with $\operatorname{rem}(P)=0$ that is~not a~loop adds nothing to~the bound $\sum_P\operatorname{rem}(P)+\nu$, so~its edge~$e_P$ must be~covered by~a~trail of~$G_1$ that is~accounted for~by~other trails of~$\mathcal{T}$.
To~describe such trails, let $\operatorname{free}(P)$ be~the number of~internal positions at~which~$P$ absorbs no~loop; these positions include all internal positions~$i$ with $a_i\ne a_{i+1}$.
The trail~$P$ has $m_P-1$ internal positions, so~$\operatorname{free}(P)=m_P-1-\operatorname{loops}(P)$, that~is, $\operatorname{rem}(P)=2\operatorname{free}(P)-2$ if~$P$ is~closed and $\operatorname{rem}(P)=2\operatorname{free}(P)+2$ if~$P$ is~open.
We~start with a~simple property of~closed trails.

\begin{lemma}\label{lem:symbol-changes}
	Let $P\in\mathcal{T}$ be~a~closed trail with $m\ge2$ edges and symbols $a_0,\dotsc,a_{m+1}$.
	Then at~least two indices $0\le i<m$ satisfy $a_i\ne a_{i+1}$.
\end{lemma}
\begin{proof}
	Since $P$ is~closed, $a_m=a_0$.
	Since $P$ has $m \ge 2$ distinct edges, the sequence $a_0,\dotsc,a_m$ is~not constant but starts and ends with the same symbol, so~$a_i\ne a_{i+1}$ for at~least two indices $0\le i<m$.
\end{proof}

The next lemma describes the closed trails with $\operatorname{rem}(P)=0$.

\begin{lemma}\label{lem:bundle-trails}
	Every closed trail $P\in\mathcal{T}$ that is~not a~loop satisfies $\operatorname{free}(P)\ge1$, that~is, $\operatorname{rem}(P)\ge0$.
	If~$\operatorname{rem}(P)=0$, then the current string of~$P$ is~$ab^ja^\ell b$ for some distinct symbols $a$ and~$b$ and some $j,\ell\in\{1,2\}$, and $P$ absorbs a~loop at~every visit to~$aa$ or~$bb$ (if~such exists).
\end{lemma}
\begin{proof}
	Let $P$ have $m\ge2$ edges and symbols $a_0,\dotsc,a_{m+1}$.
	By~\Cref{lem:symbol-changes}, at~least two indices $0\le i<m$ satisfy $a_i\ne a_{i+1}$, and all of~them except possibly $i=0$ are internal positions at~which $P$ absorbs no~loop.
	Hence, $\operatorname{free}(P)\ge1$, and if~$\operatorname{free}(P)=1$, then there are exactly two such indices, $0$ and some~$j$, and $P$ absorbs loops at~all internal positions except~$j$.
	With $a=a_0$ and $b=a_1$, this means that $a_1=\dotsb=a_j=b$ and $a_{j+1}=\dotsb=a_m=a$, and $a_{m+1}=a_1=b$.
	At~every internal position~$i$ with $a_i=a_{i+1}$, $P$ absorbs a~loop, so~$a_{i-1}\ne a_i\ne a_{i+2}$; hence, $j\le2$ and $m-j\le2$, and the current string $a_0a_1\dotsm a_{m+1}$ of~$P$ is~$ab^ja^{m-j}b$, which proves the second statement.
\end{proof}

By~\Cref{lem:bundle-trails}, every closed trail $P\in\mathcal{T}$ with $\operatorname{rem}(P)=0$ that is~not a~loop has current string $ab^ja^\ell b$ for two distinct symbols $a$ and~$b$, so~we~can group these trails by~their pair of~symbols.
Namely, for distinct symbols $a$ and~$b$, the \emph{bundle} $Z_{a,b}$ consists of~the closed trails $P\in\mathcal{T}$ with $\operatorname{rem}(P)=0$ whose symbols are $a$ and~$b$, together with the loops that these trails absorb; thus, $Z_{a,b}=Z_{b,a}$.
The trails of~$\mathcal{T}$ that are neither loops nor in~bundles, that~is, the open trails and the closed trails~$P$ with $\operatorname{rem}(P)>0$, are called \emph{connectors}.

\begin{graybox}
	Recall the graph~$G_2$ of~the running example.
	\begin{center}
		\begin{tikzpicture}[tq/.style={draw=blue!30!gray, line width=1pt}, tl/.style={draw=violet!30!gray, line width=1pt}, tb/.style={draw=orange!40!gray, line width=1pt}, tc/.style={draw=green!25!gray, line width=1pt}, tr/.style={draw=teal!45!gray, line width=1pt}, tz/.style={draw=magenta!30!gray, line width=1pt}]
			\foreach \s/\x/\y in {db/0/0, ba/1.8/0, aa/3.6/0, ac/5.4/0, ca/7.2/0, ab/1.8/1.6, eg/9/0, ge/10.8/0, ef/11.7/1.4, fg/12.6/0}
				\node[v] (\s) at (\x, \y) {\strut\ttfamily\s};
			\path (db) edge[le, tq] node[below] {\scriptsize\ttfamily dba} (ba);
			\path (ba) edge[le, tq] node[below] {\scriptsize\ttfamily baa} (aa);
			\path (aa) edge[le, tq] node[below] {\scriptsize\ttfamily aac} (ac);
			\path (aa) edge[le, tl, loop above, out=125, in=55, min distance=7mm] node[above] {\scriptsize\ttfamily aaa} (aa);
			\path (ab) edge[le, tb, bend left=30] node[right] {\scriptsize\ttfamily aba} (ba);
			\path (ba) edge[le, tb, bend left=30] node[left] {\scriptsize\ttfamily bab} (ab);
			\path (ac) edge[le, tc, bend left=35] node[above] {\scriptsize\ttfamily aca} (ca);
			\path (ca) edge[le, tc, bend left=35] node[below] {\scriptsize\ttfamily cac} (ac);
			\path (eg) edge[le, tz, bend left=35] node[above] {\scriptsize\ttfamily ege} (ge);
			\path (ge) edge[le, tz, bend left=35] node[below] {\scriptsize\ttfamily geg} (eg);
			\path (ge) edge[le, tr] node[above left] {\scriptsize\ttfamily gef} (ef);
			\path (ef) edge[le, tr] node[above right] {\scriptsize\ttfamily efg} (fg);
			\path (fg) edge[le, tr] node[below] {\scriptsize\ttfamily fge} (ge);
		\end{tikzpicture}
	\end{center}
	Let $Q$ be~the open trail~$\texttt{db}\to\texttt{ba}\to\texttt{aa}\to\texttt{ac}$, so~its internal positions are $1$ and~$2$.
	It~absorbs the loop \texttt{aaa} at~position~$2$, hence, $\operatorname{loops}(Q)=1$, $\operatorname{free}(Q)=1$, and $\operatorname{rem}(Q)=\operatorname{contr}(Q)-2\operatorname{loops}(Q)=4$.
	The closed trail $\texttt{ab}\to\texttt{ba}\to\texttt{ab}$ absorbs no~loop at~its only internal position~$1$, so~$\operatorname{free}(P)=1$ and $\operatorname{rem}(P)=2\operatorname{free}(P)-2=0$; similarly, the closed trail $\texttt{ca}\to\texttt{ac}\to\texttt{ca}$ has $\operatorname{rem}(P)=0$.
	Their current strings \texttt{abab} and~\texttt{caca} have the form $ab^ja^\ell b$ of~\Cref{lem:bundle-trails} with $j=\ell=1$, so~they form the bundles $Z_{\texttt{a},\texttt{b}}$ and~$Z_{\texttt{a},\texttt{c}}$, whereas $Q$ is~a~connector.
	In~the second component, the closed trail $Q'=\texttt{fg}\to\texttt{ge}\to\texttt{ef}\to\texttt{fg}$ absorbs no~loops at~its internal positions $1$ and~$2$, so~$\operatorname{free}(Q')=2$ and $\operatorname{rem}(Q')=2\operatorname{free}(Q')-2=2$, and $Q'$ is~a~connector; the closed trail $\texttt{eg}\to\texttt{ge}\to\texttt{eg}$ has $\operatorname{rem}(P)=0$ and forms the bundle $Z_{\texttt{e},\texttt{g}}$.
	Thus, $\sum_P\operatorname{rem}(P)=\operatorname{rem}(Q)+\operatorname{rem}(Q')=6$, and the edges of~the bundles in~$G_1$ must be~covered together with~$e_Q$ and~$e_{Q'}$, by~trails accounted for~by~$\operatorname{rem}(Q)$ and~$\operatorname{rem}(Q')$.
\end{graybox}

\begin{lemma}\label{lem:bundles}
	Besides loops, every bundle $Z_{a,b}$ contains at~most one trail that starts at~$ab$ and at~most one trail that starts at~$ba$.
	The edges of~the trails of~$Z_{a,b}$ in~$G_1$ form a~trail of~$G_1$.
	Moreover, different bundles visit disjoint sets of~vertices.
\end{lemma}
\begin{proof}
	Since no~trail of~$\mathcal{T}$ ends at~the first vertex of~another one, different closed trails of~$\mathcal{T}$ start at~different vertices, and by~\Cref{lem:bundle-trails}, every trail of~$Z_{a,b}$ that is~not a~loop starts at~$ab$ or~$ba$.
	The edges of~these trails in~$G_1$ are therefore an~edge from $a$ to~$b$, an~edge from $b$ to~$a$, or~both, and the loops of~$Z_{a,b}$ give loops at~$a$ and~$b$, so~these edges form a~trail of~$G_1$.
	By~\Cref{lem:bundle-trails}, the trails of~$Z_{a,b}$ visit only the vertices $ab$, $ba$, $aa$, and~$bb$, and they absorb the loop $aaa$ at~every visit to~$aa$.
	Hence, two bundles can share only a~vertex~$aa$, but then both would absorb the loop $aaa$.
\end{proof}

A~weak component of~$G_2$ without connectors consists of~a~single loop or~of~a~single bundle, since every loop in~a~weak component with other edges is~absorbed by~a~trail of~a~bundle, and different bundles share no~vertices.
Such a~component contains no~open trail of~$\mathcal{T}$, so~there are at~most $\nu$ such components, and the term~$\nu$ in~the bound $\sum_P\operatorname{rem}(P)+\nu$ accounts for~the one trail that covers each of~them.

Now, consider a~bundle $Z=Z_{a,b}$ in~a~weak component of~$G_2$ with a~connector.
Then some edge outside~$Z$ is~incident to~a~vertex visited by~$Z$.
The trail of~$\mathcal{T}$ that contains this edge is~not a~loop, since $Z$ absorbs the loops, and it~is~not in~another bundle by~\Cref{lem:bundles}, so~some connector visits a~vertex visited by~$Z$.
We~attach $Z$ to~one of~the connectors as~follows.

If~some connector~$Q$ starts or~ends at~a~vertex~$v$ visited by~$Z$, we attach $Z$ to~$Q$ at~this endpoint.
No~trail of~$Z$ starts at~$v$.
Indeed, every trail of~$Z$, including its loops, is~closed, and no~trail of~$\mathcal{T}$ ends at~the first vertex of~another one.
In~particular, $v$ is~neither $aa$ nor~$bb$, since by~\Cref{lem:bundle-trails}, if~$Z$ visits~$aa$, then it~contains the loop $aaa$, which starts at~$aa$, and similarly for~$bb$.
Hence, $v$ is~$ab$ or~$ba$, say, $v=ab$.
By~\Cref{lem:bundle-trails}, every trail of~$Z$ that is~not a~loop starts at~$ab$ or~$ba$, so~it~starts at~$ba$, and by~\Cref{lem:bundles}, $Z$ consists of~a~single such trail and possibly of~loops at~$aa$ and~$bb$.
Therefore, $e_Q$ together with the edges of~$Z$ forms a~trail in~$G_1$.
By~\Cref{lem:bundles}, at~most one bundle is~attached at~each endpoint of~$Q$, and $e_Q$ together with the bundles attached at~both endpoints still forms a~trail of~$G_1$.

\begin{graybox}
	In~the running example, the connector $Q=\texttt{db}\to\texttt{ba}\to\texttt{aa}\to\texttt{ac}$ ends at~$v=\texttt{ac}$, which the bundle $Z_{\texttt{a},\texttt{c}}$ visits, so~$Z_{\texttt{a},\texttt{c}}$ is~attached to~$Q$ at~this endpoint.
	This bundle consists of~the single trail $\texttt{ca}\to\texttt{ac}\to\texttt{ca}$, which starts at~\texttt{ca} (not at~$v$), and its edge \texttt{caca} in~$G_1$ goes from \texttt{c} to~\texttt{a}.
	The edge $e_Q=\texttt{dbaac}$ goes from \texttt{d} to~\texttt{c}, so~together they form the trail $\texttt{d}\to\texttt{c}\to\texttt{a}$ of~$G_1$.
\end{graybox}

Otherwise, no~connector starts or~ends at~a~vertex visited by~$Z$, and we~take a~connector~$Q$ that visits a~vertex visited by~$Z$ and attach $Z$ to~$Q$ at~an~internal position~$i$ such that $Z$ visits~$Q_i$.
The trail $Q$ absorbs no~loop at~$i$, since either $Q_i$ is~$ab$ or~$ba$, or~the loop at~$Q_i$ is~absorbed by~$Z$.
By~\Cref{lem:bundles}, different bundles are attached at~different positions.

Now, every edge of~$G_1$ outside the weak components of~$G_2$ without connectors lies in~the group of~some connector~$Q$, which consists of~$Q$, the bundles attached to~$Q$, and the loops absorbed by~$Q$ or~by~these bundles.
We~cover the edges of~each group in~$G_1$ separately, except for~some loops: a~loop $aaa$ can be~inserted into any trail of~$G_1$ that visits~$a$, so~it~needs a~trail of~its own only if~$a$ has no~other edges in~$G_1$.
A~loop $aaa$ is~\emph{isolated in~$G_1$} if~no~other edge of~$G_1$ is~incident to~$a$.
By~\Cref{lem:bundle-trails}, the loops absorbed by~bundles are not isolated in~$G_1$.

\begin{graybox}
	Note that a~loop becomes isolated when its symbol occurs only \emph{inside} the current strings.
	For example, for the strings \texttt{xyy}, \texttt{yyy}, and~\texttt{yyz}, the greedy algorithm could merge \texttt{xyy} and~\texttt{yyz} into \texttt{xyyz}, so~that the open trail $\texttt{xy}\to\texttt{yy}\to\texttt{yz}$ absorbs the loop \texttt{yyy}.
	In~$G_1$, the string \texttt{xyyz} gives an~edge from \texttt{x} to~\texttt{z}, and no~other current string starts or~ends with~\texttt{y}, so~the loop \texttt{yyy} is~isolated in~$G_1$ and needs a~trail of~its own.
\end{graybox}

For a~connector~$Q$, let $D_Q$ be~the subgraph of~$G_1$ formed by~$e_Q$, by~the edges in~$G_1$ of~the bundles attached to~$Q$, and by~the edges of~the loops isolated in~$G_1$ that $Q$ absorbs.
The next lemma shows that the remaining contribution of~$Q$ bounds the size of~a~trail cover of~$D_Q$.

\begin{lemma}\label{lem:local-cover}
	Every connector~$Q$ satisfies $\pi(D_Q)\le \operatorname{rem}(Q)$.
\end{lemma}
\begin{proof}
	Let $B_Q$ be~the set of~internal positions at~which bundles are attached to~$Q$; since $Q$ absorbs no~loop at~these positions, $|B_Q|\le \operatorname{free}(Q)$.
	Let $\lambda_Q$ be~the number of~loops isolated in~$G_1$ that $Q$ absorbs.
	Suppose that $Q$ absorbs such a~loop $aaa$ at~position~$i$, so~that $Q_{i-1}=ba$ and $Q_{i+1}=ac$ for some symbols $b\ne a$ and $c\ne a$.
	No~bundle visits~$ba$, since a~bundle that visits~$ba$ is~$Z_{a,b}$, whose edges in~$G_1$ are incident to~$a$; similarly, no~bundle visits~$ac$.
	If~$Q$ is~closed, then $i-1$ and~$i+1$ are internal positions, since otherwise $Q$ starts at~$ba$ or~$ac$, and $e_Q$ is~incident to~$a$.
	Since $c\ne a$, $Q$ absorbs no~loop at~$i+1$.
	Hence, the positions $i+1$ of~these $\lambda_Q$ loops are distinct positions without loops outside~$B_Q$, and they are internal, except that maybe $i+1=m_Q$ is~possible if~$Q$ is~open.
	Therefore,
	\begin{equation}\label{eq:bundles-and-loops}
		|B_Q|+\lambda_Q\le\begin{cases}\operatorname{free}(Q)+1&\text{if~$Q$ is~open},\\\operatorname{free}(Q)&\text{if~$Q$ is~closed}.\end{cases}
	\end{equation}
	The edge $e_Q$ and the bundles attached at~the endpoints of~$Q$ form one trail of~$G_1$, and so~does every bundle attached at~a~position in~$B_Q$, by~\Cref{lem:bundles}.
	Together with the $\lambda_Q$ isolated loops, this gives $\pi(D_Q)\le1+|B_Q|+\lambda_Q$.
	Recall that $\operatorname{rem}(Q)=2\operatorname{free}(Q)+2$ if~$Q$ is~open and $\operatorname{rem}(Q)=2\operatorname{free}(Q)-2$ if~$Q$ is~closed.
	Hence, by~\eqref{eq:bundles-and-loops}, $\pi(D_Q)\le \operatorname{free}(Q)+2\le \operatorname{rem}(Q)$ if~$Q$ is~open, and $\pi(D_Q)\le \operatorname{free}(Q)+1\le2\operatorname{free}(Q)-2=\operatorname{rem}(Q)$ if~$Q$ is~closed and $\operatorname{free}(Q)\ge3$.

	It~remains to~consider a~closed connector~$Q$ with $\operatorname{free}(Q)=2$, that~is, with $\operatorname{rem}(Q)=2$.
	By~\eqref{eq:bundles-and-loops}, $|B_Q|+\lambda_Q\le2$, and if~$|B_Q|+\lambda_Q\le1$, then $\pi(D_Q)\le1+|B_Q|+\lambda_Q\le2=\operatorname{rem}(Q)$.
	Hence, we~may assume that $|B_Q|+\lambda_Q=2$.
	We~first show that $\lambda_Q=0$.
	Suppose that $Q$ absorbs a~loop isolated in~$G_1$ at~position~$i$.
	As~shown above, $Q_{i-1}=ba$ and $Q_{i+1}=ac$ for symbols $b\ne a$ and $c\ne a$, and since $Q$ is~closed, $i-1$ and~$i+1$ are internal positions.
	The vertices $ba$ and~$ac$ consist of~two distinct symbols, so~$Q$ absorbs no~loop at~$i-1$ and~$i+1$, and since $\operatorname{free}(Q)=2$, these are the only internal positions at~which $Q$ absorbs no~loop.
	Every position in~$B_Q$ is~such a~position, and a~bundle attached at~a~position~$j$ visits~$Q_j$, but, as~shown above, no~bundle visits $ba$ or~$ac$; hence, $B_Q=\varnothing$.
	Then $\lambda_Q=2 - |B_Q| = 2$, so~$Q$ absorbs another loop isolated in~$G_1$ at~a~position $i'\ne i$.
	By~the same argument, $i'-1$ and~$i'+1$ are also internal positions at~which $Q$ absorbs no~loop, so~$\{i'-1,i'+1\}=\{i-1,i+1\}$, which implies $i'=i$, a~contradiction.

	Hence, $\lambda_Q = 0$, and $|B_Q|=2$, that~is, bundles are attached at~both internal positions of~$Q$ without loops.
	Let $Q$ have $m$ edges and symbols $a_0,\dotsc,a_{m+1}$.
	By~\Cref{lem:symbol-changes}, at~least two indices $0\le i<m$ satisfy $a_i\ne a_{i+1}$, and $Q$ absorbs no~loop at~the internal ones among them; since $\operatorname{free}(Q)=2$, there are two or~three such indices.

	Suppose that there are two such indices.
	Since $a_m=a_0$, the sequence $a_0,\dotsc,a_m$ changes from some symbol~$a$ to~another symbol~$b$ at~one of~these indices and back from $b$ to~$a$ at~the other, so~$Q$ visits $ab$ and~$ba$ at~them.
	A~bundle attached to~$Q$ at~one of~these indices visits $ab$ or~$ba$, so~this bundle is~$Z_{a,b}$, which visits both $ab$ and~$ba$.
	If~one of~the two indices is~$0$, then $Z_{a,b}$ visits $Q_0$, so~it~is~attached at~an~endpoint and not at~the other index, and thus $|B_Q|\le \operatorname{free}(Q)-1$.
	Otherwise, both indices are internal positions without loops, so~they form~$B_Q$, but $Z_{a,b}$ is~attached at~only one position.
	Both cases contradict $|B_Q|=2$.

	Hence, there are three such indices, and one of~them is~$0$, since $\operatorname{free}(Q)=2$.
	Let these indices be~$0<j<k<m$, and let $a=a_0$, $b=a_1$, and $c=a_{j+1}$.
	Then $a_1=\dotsb=a_j=b$, $a_{j+1}=\dotsb=a_k=c$, and $a_{k+1}=\dotsb=a_m=a$, where the last equality holds since $a_m=a_0$.
	Hence, $Q$ visits $ab$, $bc$, and~$ca$ at~the positions $0$, $j$, and~$k$.
	The bundles attached at~the positions $j$ and~$k$ visit $bc$ and~$ca$, so~they are $Z_{b,c}$ and~$Z_{c,a}$.
	
	We~now list the edges of~$D_Q$ that are not loops.
	\begin{itemize}
		\item The edge~$e_Q$, which goes from $a_0=a$ to~$a_{m+1}=a_1=b$.
		\item An~edge from $b$ to~$a$ if~a~bundle is~attached at~an~endpoint of~$Q$: since $Q$ starts and ends at~$ab$, this bundle is~$Z_{a,b}$, and as~shown above, its only trail that is~not a~loop starts at~$ba$.
		\item At~most one edge from $b$ to~$c$ and at~most one edge from $c$ to~$b$, which come from $Z_{b,c}$, and similarly, at~most one edge from $c$ to~$a$ and at~most one edge from $a$ to~$c$, which come from $Z_{c,a}$, by~\Cref{lem:bundles}.
	\end{itemize}
	Since $\lambda_Q=0$, there are no~other edges, so, ignoring loops, every symbol has in-degree and out-degree at~most two in~$D_Q$, that~is, $-2\le\deg^+(v)-\deg^-(v)\le2$ for every vertex~$v$.
	These differences sum to~zero, so~if~$\operatorname{exc}(D_Q)\ge3$, then at~least two vertices have positive differences and at~least two other vertices have negative differences, which is~impossible, since $D_Q$ has only the three vertices $a$, $b$, and~$c$.
	Hence, $\operatorname{exc}(D_Q)\le2$.
	Finally, $D_Q$ is~weakly connected, since $e_Q$ joins $a$ and~$b$, and $Z_{b,c}$ joins $b$ and~$c$, so~\Cref{lem:trail-number} gives $\pi(D_Q)\le2=\operatorname{rem}(Q)$.
\end{proof}

\begin{graybox}
	Recall the graphs $G_2$ and~$G_1$ of~the running example.
	\begin{center}
		\begin{tikzpicture}[tq/.style={draw=blue!30!gray, line width=1pt}, tl/.style={draw=violet!30!gray, line width=1pt}, tb/.style={draw=orange!40!gray, line width=1pt}, tc/.style={draw=green!25!gray, line width=1pt}, tr/.style={draw=teal!45!gray, line width=1pt}, tz/.style={draw=magenta!30!gray, line width=1pt}]
			\foreach \s/\x/\y in {db/0/0, ba/1.8/0, aa/3.6/0, ac/5.4/0, ca/7.2/0, ab/1.8/1.6, eg/9/0, ge/10.8/0, ef/11.7/1.4, fg/12.6/0}
				\node[v] (\s) at (\x, \y) {\strut\ttfamily\s};
			\path (db) edge[le, tq] node[below] {\scriptsize\ttfamily dba} (ba);
			\path (ba) edge[le, tq] node[below] {\scriptsize\ttfamily baa} (aa);
			\path (aa) edge[le, tq] node[below] {\scriptsize\ttfamily aac} (ac);
			\path (aa) edge[le, tl, loop above, out=125, in=55, min distance=7mm] node[above] {\scriptsize\ttfamily aaa} (aa);
			\path (ab) edge[le, tb, bend left=30] node[right] {\scriptsize\ttfamily aba} (ba);
			\path (ba) edge[le, tb, bend left=30] node[left] {\scriptsize\ttfamily bab} (ab);
			\path (ac) edge[le, tc, bend left=35] node[above] {\scriptsize\ttfamily aca} (ca);
			\path (ca) edge[le, tc, bend left=35] node[below] {\scriptsize\ttfamily cac} (ac);
			\path (eg) edge[le, tz, bend left=35] node[above] {\scriptsize\ttfamily ege} (ge);
			\path (ge) edge[le, tz, bend left=35] node[below] {\scriptsize\ttfamily geg} (eg);
			\path (ge) edge[le, tr] node[above left] {\scriptsize\ttfamily gef} (ef);
			\path (ef) edge[le, tr] node[above right] {\scriptsize\ttfamily efg} (fg);
			\path (fg) edge[le, tr] node[below] {\scriptsize\ttfamily fge} (ge);
			\foreach \s/\x in {d/1.2, c/2.8, a/4.4, b/6, e/9.4, g/11, f/12.6}
				\node[v] (\s) at (\x, -2.1) {\strut\ttfamily\s};
			\path (d) edge[le, tq] node[below] {\scriptsize\ttfamily dbaac} (c);
			\path (c) edge[le, tc] node[below] {\scriptsize\ttfamily caca} (a);
			\path (a) edge[le, tb] node[below] {\scriptsize\ttfamily abab} (b);
			\path (a) edge[le, tl, loop above, out=125, in=55, min distance=7mm] node[above] {\scriptsize\ttfamily aaa} (a);
			\path (e) edge[le, tz] node[below] {\scriptsize\ttfamily egeg} (g);
			\path (f) edge[le, tr] node[below] {\scriptsize\ttfamily fgefg} (g);
		\end{tikzpicture}
	\end{center}
	In~this example $Z_{\texttt{a},\texttt{b}}$ is~attached to~$Q$ at~the internal position~$1$, where $Q$ visits \texttt{ba}.
	The loop \texttt{aaa} is~not isolated in~$G_1$, so~$D_Q$ consists of~the edges \texttt{dbaac}, \texttt{caca}, and~\texttt{abab} of~$G_1$.
	The proof of~\Cref{lem:local-cover} covers them by~the trail $\texttt{d}\to\texttt{c}\to\texttt{a}$ and by~the edge $\texttt{a}\to\texttt{b}$ of~the bundle attached at~position~$1$, so~$\pi(D_Q)\le2\le \operatorname{rem}(Q)$.
	In~the second component, no~connector starts or~ends at~\texttt{eg} or~\texttt{ge}, so~$Z_{\texttt{e},\texttt{g}}$ is~attached to~$Q'$ at~the internal position~$1$, where $Q'$ visits \texttt{ge}.
	Then $D_{Q'}$ consists of~the edges \texttt{fgefg} and~\texttt{egeg} of~$G_1$, which both end at~\texttt{g}, so~$\pi(D_{Q'})=2=\operatorname{rem}(Q')$, and \Cref{lem:local-cover} holds with equality.
	Finally, as~in~the proof below, the loop \texttt{aaa} is~inserted at~\texttt{a}, which gives the trail $\texttt{d}\to\texttt{c}\to\texttt{a}\to\texttt{a}\to\texttt{b}$; together with the two trails that cover~$D_{Q'}$, it~covers~$G_1$ with three trails, whereas $\sum_Q\operatorname{rem}(Q)+\nu=7$.
\end{graybox}

\begin{proof}[Proof of~\Cref{lem:length-three-main}]
	We~cover the edges of~$G_1$ as~follows.
	For every weak component of~$G_2$ without connectors, we~take the trail of~$G_1$ formed by~its edges in~$G_1$, and for every connector~$Q$, we~take a~minimum trail cover of~$D_Q$.
	These trails are edge-disjoint, since every bundle is~attached to~at~most one connector and every loop is~absorbed by~at~most one trail, and they cover all edges of~$G_1$ except the loops absorbed by~connectors that are not isolated in~$G_1$.
	Such a~loop $aaa$ is~the only uncovered edge at~$a$, and another edge of~$G_1$ is~incident to~$a$, so~the loop can be~inserted into a~trail that visits~$a$.
	The weak components of~$G_2$ without connectors contain no~open trails, so~there are at~most~$\nu$ of~them.
	Therefore \Cref{lem:local-cover,lem:shares} give
	\[\pi(G_1)\le\sum_Q\operatorname{rem}(Q)+\nu=\sum_P\operatorname{rem}(P)+\nu\le4\nu+\sum_{P\in\mathcal{T}}\operatorname{contr}(P),\]where the first sum is~over connectors~$Q$ and the second over the trails $P\in\mathcal{T}$ that are not loops; these two sums are equal, since $\operatorname{rem}(P)=0$ for every trail~$P$ in~a~bundle.
\end{proof}
 \section{Lower Bound for~Strings of~Length at~Least Six}
\label{sec:fixed-length-lower-bound}

In~this section, we~show that $\rho_k\ge2$ for every $k\ge6$, where $\rho_k$, defined in~\Cref{sec:length-three}, is~the approximation ratio of~the greedy algorithm on~strings of~length~$k$.
Thus, whereas the greedy algorithm is~a~$9/5$-approximation on~strings of~length three by~\Cref{thm:length-three}, its approximation ratio on~strings of~any fixed length $k\ge6$ is~at~least~$2$.

\begin{theorem}[Formal version of~\Cref{thm:fixed-length-lower-bound-informal}]\label{thm:fixed-length-lower-bound}
	For all integers $k\ge6$ and $\ell\ge1$, there is~a~set of~$(2k-3)\ell$ strings of~length~$k$ that has a~common superstring of~length $(2k-3)\ell+k-1$, whereas the greedy algorithm may output a~common superstring of~length $(4k-6)\ell$.
	Consequently, $\rho_k\ge2$ for every $k\ge6$.
\end{theorem}

As~in~\Cref{prop:length-three-tight}, ``may'' means that the greedy algorithm produces this output for some way of~breaking ties.
The second statement follows from the first one, since $(4k-6)\ell/((2k-3)\ell+k-1)$ tends to~$2$ as~$\ell\to\infty$.

Before giving the details, we~explain the idea of~the construction.
As~in~\Cref{sec:length-three}, consider the graph whose vertices are strings of~length $k-1$ and in~which every input string is~an~edge from its length-$(k-1)$ prefix to~its length-$(k-1)$ suffix; it~is~a~subgraph of~the de~Bruijn graph~\cite{DeBrujin46}.
Two input strings overlap by~$k-1$ symbols exactly when the first one ends at~the vertex where the second one starts.
The input consists of~$\ell$ blocks, and the strings of~each block form two cycles with $k$ and $k-3$ edges that share a~vertex.
Moreover, the cycles with $k$ edges of~all blocks pass through one common vertex.
Hence there is~a~common superstring that walks along both cycles of~every block and passes from one block to~the next through this common vertex, so~that consecutive strings in~it overlap by~$k-1$ symbols, and every block adds only $2k-3$ symbols to~it, one per string.
The greedy algorithm, however, may merge the strings of~each of~the two cycles of~a~block separately, and the two resulting strings overlap by~a~single symbol, so~every block adds $4k-6$ symbols to~its output.

\begin{graybox}
	For $k=6$ and $\ell=1$, the input consists of~the six rotations of~$\texttt{abaabx}$ and the three strings $\texttt{abxabx}$, $\texttt{bxabxa}$, and~$\texttt{xabxab}$, which form the following two cycles.
	\begin{center}
		\begin{tikzpicture}
			\foreach \s/\x/\y in {abxab/0/0, bxaba/-1.6/1.2, xabaa/-4.4/1.2, abaab/-6/0, baabx/-4.4/-1.2, aabxa/-1.6/-1.2, bxabx/2.8/1.2, xabxa/2.8/-1.2}
				\node[v] (\s) at (\x, \y) {\strut\ttfamily\s};
			\path (abaab) edge[le] node[below left] {\scriptsize\ttfamily abaabx} (baabx);
			\path (baabx) edge[le] node[below] {\scriptsize\ttfamily baabxa} (aabxa);
			\path (aabxa) edge[le] node[below right] {\scriptsize\ttfamily aabxab} (abxab);
			\path (abxab) edge[le] node[above right] {\scriptsize\ttfamily abxaba} (bxaba);
			\path (bxaba) edge[le] node[above] {\scriptsize\ttfamily bxabaa} (xabaa);
			\path (xabaa) edge[le] node[above left] {\scriptsize\ttfamily xabaab} (abaab);
			\path (abxab) edge[le] node[below right] {\scriptsize\ttfamily abxabx} (bxabx);
			\path (bxabx) edge[le] node[right] {\scriptsize\ttfamily bxabxa} (xabxa);
			\path (xabxa) edge[le] node[above right] {\scriptsize\ttfamily xabxab} (abxab);
		\end{tikzpicture}
	\end{center}
	The common superstring $\texttt{abaabxabxabaab}$ of~length~$14$ walks from $\texttt{abaab}$ along the left cycle to~$\texttt{abxab}$, around the right cycle, and back to~$\texttt{abaab}$.
	The greedy algorithm may instead merge the rotations of~$\texttt{abaabx}$ into $\texttt{xabaabxabaa}$ and the other three strings into $\texttt{bxabxabx}$, which overlap by~one symbol, and output $\texttt{bxabxabxabaabxabaa}$ of~length~$18$.
\end{graybox}

\begin{proof}[Proof of~\Cref{thm:fixed-length-lower-bound}]
	Let $\texttt{a},\texttt{b},\texttt{c},\texttt{x}_1,\dotsc,\texttt{x}_\ell$ be~distinct symbols, and let $y=\texttt{c}^{k-6}$, which is~empty for $k=6$.
	For every $i\in[\ell]$, let
	\[u_i=\texttt{abaab}\texttt{x}_iy\qquad\text{and}\qquad v_i=\texttt{ab}\texttt{x}_iy,\]
	so~that $|u_i|=k$ and $|v_i|=k-3$.
	The $i$-th block consists of~the strings of~$\operatorname{CSub}_k(u_i)$, that~is, the rotations of~$u_i$, and of~the strings of~$\operatorname{CSub}_k(v_i)$, and the input~$\mathcal{S}$ is~the union of~the $\ell$ blocks.

	We~first show that $|\mathcal{S}|=(2k-3)\ell$.
	Since $\texttt{x}_i$ occurs in~$u_i$ and~$v_i$ only once, the position of~the first~$\texttt{x}_i$ distinguishes the $k$ rotations of~$u_i$, as~well as~the length-$k$ substrings of~$v_i^\infty$ that start at~$k-3$ consecutive positions, which form $\operatorname{CSub}_k(v_i)$.
	Moreover, every rotation of~$u_i$ contains three~$\texttt{a}$'s, whereas every string of~$\operatorname{CSub}_k(v_i)$ contains at~most two, since any three occurrences of~$\texttt{a}$ in~$v_i^\infty$ span at~least $2k-5>k$ symbols.
	Finally, all strings of~the $i$-th block contain~$\texttt{x}_i$, so~different blocks are disjoint.

	The string
	\[h=y\texttt{abaab}\texttt{x}_1y\texttt{ab}\texttt{x}_1y\texttt{abaab}\dotsm\texttt{x}_\ell y\texttt{ab}\texttt{x}_\ell y\texttt{abaab}\]
	has length $(2k-3)\ell+k-1$, and it~is~a~common superstring of~$\mathcal{S}$.
	Indeed, for every~$i$, the string~$h$ contains $y\texttt{abaab}\texttt{x}_iy\texttt{ab}\texttt{x}_iy\texttt{abaab}$.
	The suffix $\texttt{ab}\texttt{x}_iy\texttt{abaab}$ of~this string contains the rotations $\texttt{ab}\texttt{x}_iy\texttt{aba}$, $\texttt{b}\texttt{x}_iy\texttt{abaa}$, and~$\texttt{x}_iy\texttt{abaab}$ of~$u_i$, and its prefix $y\texttt{abaab}\texttt{x}_iy\texttt{ab}$ is~a~substring of~$u_i^\infty$ of~length $2k-4$ whose $k-3$ length-$k$ substrings are the other rotations of~$u_i$.
	The substring $\texttt{ab}\texttt{x}_iy\texttt{ab}\texttt{x}_iy\texttt{ab}$ of~$h$ is~a~substring of~$v_i^\infty$ of~length $2k-4=|v_i|+k-1$, so~its length-$k$ substrings form $\operatorname{CSub}_k(v_i)$.
	Moreover, $h$ is~a~shortest common superstring of~$\mathcal{S}$: the strings of~$\mathcal{S}$ are distinct and have length~$k$, and hence its common superstring is of length at~least $|\mathcal{S}|+k-1=|h|$.

	We~now describe a~run of~the greedy algorithm.
	For every~$i$, let
	\[u'_i=\texttt{x}_iy\texttt{abaab}\texttt{x}_iy\texttt{abaa}\qquad\text{and}\qquad v'_i=\texttt{b}\texttt{x}_iy\texttt{ab}\texttt{x}_iy\texttt{ab}\texttt{x}_i.\]
	These are substrings of~$u_i^\infty$ and~$v_i^\infty$ of~lengths $2k-1=|u_i|+k-1$ and $2k-4=|v_i|+k-1$, so~their length-$k$ substrings form $\operatorname{CSub}_k(u_i)$ and~$\operatorname{CSub}_k(v_i)$, respectively.
	Consecutive length-$k$ substrings of~a~string overlap by~$k-1$ symbols, and by~\Cref{lem:greedy-ends}, no~two current strings of~the greedy algorithm overlap by~more than $k-1$ symbols.
	Hence, the greedy algorithm may start by~merging the length-$k$ substrings of~every $u'_i$ and~$v'_i$ from left to~right, which gives the $2\ell$ current strings $u'_1,v'_1,\dotsc,u'_\ell,v'_\ell$.

	We~claim that $\operatorname{ov}(v'_i,u'_i)=\texttt{x}_i$ for every~$i$, and that every other overlap of~two distinct strings among $u'_1,v'_1,\dotsc,u'_\ell,v'_\ell$ is~empty.
	Note that $u'_i$ starts with~$\texttt{x}_i$, ends with~$\texttt{aa}$, and contains~$\texttt{x}_i$ only at~positions $1$ and~$k+1$, whereas $v'_i$ starts with~$\texttt{b}$, ends with~$\texttt{x}_i$, and does not contain~$\texttt{aa}$.
	In~particular, $\operatorname{ov}(v'_i,u'_i)$ is~nonempty.
	Let $s\in\{u'_i,v'_i\}$ and $t\in\{u'_j,v'_j\}$ be~distinct strings such that $\operatorname{ov}(s,t)$ is~nonempty.
	If~$\operatorname{ov}(s,t)$ ends with~$\texttt{a}$, then $s=u'_i$, and $\operatorname{ov}(s,t)\ne\texttt{a}$, since $t$ starts with~$\texttt{x}_j$ or~$\texttt{b}$.
	Then $\operatorname{ov}(s,t)$ ends with~$\texttt{aa}$, so~$t=u'_j$, and $\operatorname{ov}(s,t)$ is~a~suffix of~$u'_i$ that starts with~$\texttt{x}_j$, so~$j=i$ and $t=s$, a~contradiction.
	Otherwise, $s=v'_i$ and $\operatorname{ov}(s,t)$ ends with~$\texttt{x}_i$, so~$t$ contains~$\texttt{x}_i$, that~is, $t=u'_i$.
	Then $|\!\operatorname{ov}(s,t)|$ is~$1$ or~$k+1$, but in~the latter case, $\operatorname{ov}(s,t)$ would contain~$\texttt{aa}$, unlike~$v'_i$.
	This proves the claim.

	Hence, the greedy algorithm continues by~merging $v'_i$ with~$u'_i$ for every~$i$.
	By~\Cref{lem:greedy-ends}, these merges do~not change the other overlaps, so~the resulting $\ell$ strings do~not overlap, and the greedy algorithm concatenates them.
	Thus, it~outputs a~common superstring of~length
	\[\sum_{i\in[\ell]}(|v'_i|+|u'_i|-1)=(4k-6)\ell.\qedhere\]
\end{proof}
 \section*{AI methodology}\label{sec:disclaimer}
We~follow the recommendations of~the working group report~\cite{AITCS26} and~use its template for~this section.

\begin{flushleft}
	\renewcommand{\arraystretch}{1.3}
	\begin{tabular}{@{}l@{\qquad}c@{}}
		\textbf{Activity} & \textbf{AI used} \\
		Asking the research question & No \\
		Coming up with the approach & Yes \\
		Proof development & Yes \\
		Writing and exposition & Yes \\
		Checking for bugs & Yes \\
		Other supporting tasks & Yes \\
	\end{tabular}
\end{flushleft}

\paragraph{Details (elaborating on the points above).}
We~used GPT-5.6~Sol Pro to~explore different approaches to~prove and~refute the greedy conjecture.
While we~were trying to~formulate a~self-reduction for SCS, the agent suggested a~proof of~the inequality in~\Cref{thm:periodic-budget}; it~was not directly related to~the discussed question, but the agent had found a~simpler route to~the statement we~tried to~prove.
The initial proof was complicated, but working together we~simplified it~significantly.
Then we~used the agent to~do mechanical work to~construct the example in~\Cref{prop:black-box-tightness}.

The first counterexample to~the Greedy Conjecture and a series of improved lower bounds in~\Cref{thm:greedy-lower-bound} were obtained using GPT-5.6~Sol and GPT-6~Astra.
After the initial counterexample with the lower bound $9/4$, we noticed that many constructions---including it---use cactus graphs that are Eulerian. From this observation, we obtained improvements to $7/3$ and then $19/8$ with the help of GPT-6~Astra. Afterwards, the progress was stuck for a while, until we created a repository containing all the knowledge from the individual sessions, using which Claude~Opus-5.5 finally discovered the construction yielding the lower bound $12/5$.

We~also used GPT-5.6~Sol Ultra to~find all typos in~the paper.
We~have written all~the arguments by~ourselves, checked for~correctness and originality, and assume full responsibility for~the manuscript.
 \bibliographystyle{alpha}
\bibliography{refs}
\end{document}